\documentclass[a4paper,leqno,10pt]{amsart}

\usepackage{epsf,graphicx,epsfig}
\usepackage{amscd}
\usepackage{amsmath,latexsym,amssymb,amsthm}
\usepackage[nospace,noadjust]{cite}
\usepackage{textcomp}
\usepackage{setspace,cite}
\usepackage{lscape,fancyhdr,fancybox}
\usepackage{stmaryrd}
\usepackage[all,cmtip]{xy}
\usepackage{tikz}
\usepackage{cancel}
\usetikzlibrary{shapes,arrows,decorations.markings}
\usepackage{graphicx}

\usepackage{color}
\usepackage{url}
\usepackage{enumerate}
\usepackage[mathscr]{euscript}
  \theoremstyle{plain}

\swapnumbers
    \newtheorem{thm}{Theorem}[section]
    \newtheorem{prop}[thm]{Proposition}
     
   \newtheorem{lemma}[thm]{Lemma}

    \newtheorem{subsec}[thm]{}
\theoremstyle{definition}
    \newtheorem{defn}[thm]{Definition}
        \newtheorem{remark}[thm]{Remark}
    \newtheorem{exam}[thm]{Example}

\theoremstyle{remark}

\newcommand{\largewedge}{\mbox{\Large $\wedge$}}

\title{}
\author{}
\date{}
\usepackage{amssymb}

\usepackage{hyperref}
\hypersetup{
	colorlinks,
	citecolor=blue,
	filecolor=black,
	linkcolor=blue,
	urlcolor=black
}

\begin{document}
\title{Gauge transformations of Jacobi structures and contact groupoids}

\author{Apurba Das}
\email{apurbadas348@gmail.com}
\address{Stat-Math Unit,
Indian Statistical Institute, Kolkata 700108,
West Bengal, India.}


%\maketitle
\begin{abstract}
We define gauge transformations of Jacobi structures on a manifold.
This is related to gauge transformations of Poisson structures via the Poissonization.
Finally, we discuss how the contact structure of a contact groupoid is effected by a gauge transformation of the Jacobi structure on its base.
\end{abstract}

\noindent

\thispagestyle{empty}

\maketitle

%\tableofcontents

\vspace{0.2cm}
\section{Introduction}
The notion of gauge transformations of Poisson structures associated with certain closed $2$-forms was introduced by {\v S}evera and Weinstein \cite{sev-wein} 
in connection with Poisson-sigma models. Gauge transformations of Poisson structures also arise in some quantization problems \cite{jurco-schupp-wess}. Roughly, a gauge transformation of a given Poisson structure 
modify its leafwise symplectic forms by means of
the pullback of a globally defined $2$-form. Gauge equivalent Poisson structures share many important properties, namely, they gives rise to same singular foliation
on the manifold, and corresponds to isomorphic Lie algebroid structures on the cotangent bundle. Gauge transformations of Poisson structures was further studied
by Bursztyn and Radko \cite{burs, burs-radko} from the perspective of symplectic groupoids. They also provide a relationship between gauge transformations
and Xu's Morita equivalence of Poisson manifolds.

The most natural framework to study gauge transformations of Poisson structures is that of Dirac structure. Gauge transformations have also been studied in the context of multiplicative Poisson and Dirac structures on a Lie groupoid \cite{burs , ortiz-thesis}. Recently, the present author introduced gauge transformations of Nambu-Poisson structures and showed that these transformations
commute with the reduction procedure \cite{das}.

Our main objective of this paper is the notion of Jacobi manifold, introduced by Lichnerowicz \cite{lich}. A Jacobi structure on a smooth manifold $M$ consists of a pair $(\pi,E)$ of a bivector field
$\pi \in \Gamma(\largewedge^2 TM)$ and a vector field $E \in \Gamma (TM)$ satisfying certain conditions. 
%Kirillov's local Lie algebra structures on a line bundle are same as Jacobi structures when the line bundle is trivial \cite{kirillov}. 
Jacobi structures include symplectic, Poisson, contact and locally conformal symplectic (l.c.s.) structures \cite{leon-lopez-marr-padron}.
A Jacobi structure $(\pi,E)$ on $M$ defines a bundle map
$$ (\pi,E)^\sharp : T^*M \times \mathbb{R} \rightarrow TM \times \mathbb{R}, ~ (\alpha, g) \mapsto (\pi^\sharp \alpha + g E, ~ - \langle \alpha, E \rangle),$$
for $(\alpha, g) \in \Gamma (T^*M \times \mathbb{R})$.
One might expect that the natural framework to study gauge transformations of Jacobi structures is the notion of Dirac-Jacobi structure (also called $\mathcal{E}^1(M)$-Dirac structure) introduced by Wade \cite{wade1}.
A Dirac-Jacobi structure on $M$ consists of a subbundle $L \subset \mathcal{E}^1(M) = (TM \times \mathbb{R}) \oplus (T^*M \times \mathbb{R})$ satisfying certain maximally isotropic and integrability
condition (Definition \ref{dj-defn}). Then $L$ inherits the structure of a Lie algebroid and there is a distinguised $1$-cocycle of this Lie algebroid.
The graph of the bundle map $(\pi,E)^\sharp$ associated to a Jacobi structure defines a Dirac-Jacobi structure. Hence, the $1$-jet bundle $T^*M \times \mathbb{R}$ of a Jacobi manifold $M$ carries a Lie algebroid structure
by identifying this bundle with the graph of $(\pi, E)^\sharp$.

In section \ref{sec-3}, we define an action $\tau : \Omega^1(M) \times \mathbb{DJ} (M) \rightarrow \mathbb{DJ} (M), ~ (B, L) \mapsto \tau_B (L)$ of the abelian group
$\Omega^1(M)$ on the space $\mathbb{DJ} (M)$ of all Dirac-Jacobi structures on $M$. When the Dirac-Jacobi structure $L$ comes from the graph of a Jacobi structure $(\pi,E)$,
then for any $B$, the Dirac-Jacobi structure $\tau_B (L)$ need not be the graph of another Jacobi structure. This amounts to the invertibility of a certain map
and in this case, the new Jacobi structure (denoted by $(\pi_B, E_B)$ or $\tau_B (\pi,E)$) on $M$ is called the gauge transformation of $(\pi,E)$ associated to the
$1$-form $B$. We prove that gauge equivalent Jacobi structures on $M$ gives rise to isomorphic Lie algebroid structures
on $T^*M \times \mathbb{R}$ (cf. Proposition \ref{gauge-iso-lie}). As a remark, we get that gauge equivalent Jacobi structures gives rise to isomorphic Lichnerowicz-Jacobi cohomology.
We show that gauge transformations of contact structures are contact and gauge transformations of l.c.s. structures are l.c.s. (cf. Remarks \ref{lcs-contact-remark}, \ref{contact-remark-1}).
Moreover, any two contact structures on a manifold are gauge equivalent (cf. Remark \ref{contact-remark-2}).

In section \ref{sec-poissonization}, we show that our gauge transformations of Dirac-Jacobi structures is related to the gauge transformations of Dirac structures via the Diracization
process (cf. Proposition \ref{gauge-commute-diracization}). In the particular case, it shows the relation between gauge transformations of Jacobi structures and gauge transformations of Poisson structures (cf. Proposition \ref{gauge-commute-poissonization}).

%A Jacobi structure on $M$ is said to be integrable if there exists a contact groupoid $(G \rightrightarrows M, \eta, \sigma)$ over it \cite{cra-zhu}. 
A contact groupoid is a Lie groupoid $G \rightrightarrows M$ together with a contact $1$-form $\eta \in \Omega^1(G)$ and a multiplicative function $\sigma \in C^\infty(G)$ 
that satisfies certain multiplicativity condition.
Contact groupoids correspond to Jacobi structures on its base. In section \ref{sec-contact-groupoid}, we discuss
how the contact structure of a contact groupoid is effected by a gauge transformation of the Jacobi structure on its base (cf. Theorem \ref{final-thm}).

In section \ref{sec-6}, we deal with gauge transformations of multiplicative Jacobi structures on a Lie groupoid. Namely, we define gauge transformations of Jacobi groupoids and generalized Lie bialgebroids. 

Generalized contact structures are odd analouge of generalized complex structures and generalization of contact structures  \cite{wade2} (see also \cite{vitag-wade, jonas-luca}). Motivated from the notion of $B$-field transformation
of generalized complex structures, in section \ref{sec-7}, we define a similar notion for generalized contact structures. We also noticed that $B$-field
transformation of contact structures (considered as generalized contact structures) need not be contact. Hence, this notion is different than
gauge transformations of contact structures (considered as Jacobi structures).

\medskip

\noindent {\em Conclusions.} We remark that the line bundle approach of Jacobi structures was introduced by Kirillov to study locally conformal Jacobi structures \cite{kirillov}
(see also \cite{marle, bgg}). More precisely, a local Lie algebra structure on $M$ consists of a line bundle $L$ over $M$ together with a Lie bracket
$$\{ -,-\} : \Gamma L \times \Gamma L \rightarrow \Gamma L$$
on the space of sections of $L$, which is local in the sense that, for $u, v \in \Gamma L$ supported in some open set $U \subset M$, the bracket $\{u,v\}$ is supported
in $U$ as well. Kirillov's local Lie algebras are same as Jacobi structures when the line bundle $L$ is trivial. We refer \cite{cra-sal, vitagliano, vitag-wade, jonas-luca} for more details on the
line bundle approach of contact structures, contact groupoids, Dirac-Jacobi structures and generalized contact structures. Finally, we remark that the contents of the present paper can also be discussed in the line bundle framework.

We assume that the reader is familiar with some basics of Lie groupoids and Lie algebroids. See \cite{mac-xu} for details.
Given a Lie groupoid $G \rightrightarrows M$, the source map and the target map are denoted by $\alpha$ and $\beta$, respectively. The space of composable arrows
is defined by $G^{(2)} = \{ (g,h) \in G \times G~|~ \alpha (g) = \beta (h) \}.$
We denote the de Rham differential of a manifold by $d$.

\medskip

\section{Jacobi structures}\label{sec-2}
In this section, we recall some basic preliminaries on Jacobi and Dirac-Jacobi structures on a manifold \cite{igl-marr,kirillov,leon-lopez-marr-padron,wade1, nunes-gallardo}.
\begin{defn}
 Let $M$ be a smooth manifold. A Jacobi structure on $M$ consists of a pair $(\pi, E)$ of a bivector field $\pi \in \Gamma (\largewedge^2 TM)$ and a vector field
$E \in \Gamma (TM)$ such that
$$ [\pi, \pi] = 2 E \wedge \pi  ~~\text{   and  }~~  [E, \pi] = 0 ,$$
where $[~,~]$ denotes the Schouten bracket on multivector fields.
\end{defn}

A Jacobi manifold is a manifold equipped with a Jacobi structure as above. If $E = 0$ then a Jacobi structure is nothing but a Poisson structure.
Moreover, given a Jacobi structure $(\pi, E)$ on $M$ the product manifold $M \times \mathbb{R}$ carries a Poisson structure whose Poisson
bivector is given by
\begin{align}\label{pois-formula}
\widetilde{(\pi,E)} = e^{-t} \big( \pi + \frac{\partial}{\partial t} \wedge E   \big).
\end{align}
This is called the Poissonization of the Jacobi structure $(\pi,E)$ on $M$.

It is important to note that given a Jacobi structure $(\pi, E)$ on $M$, there is a bundle map $(\pi, E)^\sharp : T^*M \times \mathbb{R} \rightarrow TM \times \mathbb{R}$
given by
\begin{align}\label{sharp-map}
 (\pi, E)^\sharp (\alpha, g) = (\pi^\sharp \alpha + g E, - \langle \alpha, E \rangle), ~ \text{ for } (\alpha, g) \in \Gamma (T^*M \times \mathbb{R}).
\end{align}
The set of all hamiltonian vector fields $X_h := pr_1 \circ (\pi, E)^\sharp (d h, h) = \pi^\sharp (d h) + h E,~ h \in C^\infty(M)$,
generates a distribution $\mathcal{D}$ on $M$, called the characteristic distribution. 
Here $pr_1 : TM \times \mathbb{R} \rightarrow TM$ denotes the projection onto the first factor.

\begin{exam}
 A contact manifold is a smooth manifold $M^{2n+1}$ together with a $1$-form $\eta \in \Omega^1(M)$ such that $\eta \wedge (d \eta)^n \neq 0$ at every point.

Given a contact manifold $(M, \eta)$ there exists an isomorphism of $C^\infty(M)$-modules
$$ \flat_\eta : \Gamma (TM) \rightarrow \Gamma (T^*M), ~ X \mapsto i_X d \eta + \eta (X) \eta .$$
The corresponding Jacobi structure $(\pi, E)$ is given by
$$ \pi (\alpha, \beta) = d \eta~ (\flat_\eta^{-1} (\alpha), \flat_\eta^{-1} (\beta) )   ~ \text{ and } ~ E = \flat_\eta^{-1} (\eta),~~ \text{ for } \alpha, \beta \in \Omega^1(M).$$
In this case, the induced bundle map $(\pi, E)^\sharp : T^*M \times \mathbb{R} \rightarrow TM \times \mathbb{R}$ is invertible with inverse
\begin{align}\label{contact-inverse}
((\pi, E)^\sharp)^{-1} (X, f) = (- i_X d \eta - f \eta, ~\eta (X)), ~~\text{ for } (X, f) \in \Gamma (TM \times \mathbb{R}).
\end{align}

Conversely, a Jacobi structure $(\pi, E)$ on $M$ is induced from a contact structure on $M$ if the bundle map $(\pi, E)^\sharp$ is invertible.
\end{exam}

\begin{exam}\label{exam-lcs}
 A locally conformal symplectic (l.c.s.) manifold is a smooth manifold $M^{2n}$ together with a non-degenerate $2$-form $\omega \in \Omega^2(M)$ with the property
that for each $x \in M$ there is an open neighbourhood $U_x$ of $x$ and a function $f \in C^\infty (U_x)$ such that $(U_x, e^{-f} \omega)$ is a symplectic manifold.
Alternatively, a l.c.s. manifold is a manifold $M^{2n}$ together with a non-degenerate $2$-form $\omega \in \Omega^2(M)$ and a closed $1$-form $\theta \in \Omega^1(M)$
such that
$$ d \omega = \theta \wedge \omega$$
(see \cite{vaisman} for more details). Given a l.c.s. manifold $(M^{2n}, \omega, \theta)$, one can define a Jacobi structure $(\pi, E)$ on $M$ by
$$ \pi (\alpha, \beta) = \omega (\flat^{-1} (\alpha), \flat^{-1} (\beta)) ~~~ \text{  ~ and ~} E = \flat^{-1} (\theta),$$
where $\flat : \Gamma (TM) \rightarrow \Gamma (T^*M), ~ X \mapsto i_X \omega$ is the isomorphism of $C^\infty(M)$-modules.
\end{exam}
\begin{remark}
The Poissonization of a contact structure $\eta$ on $M$ is given by a symplectic structure $\widetilde{\eta} = e^t (pr_1^*d \eta + dt \wedge pr_1^* \eta)$ on $M \times \mathbb{R}$.

If $(\omega, \theta)$ defines a l.c.s. structure on $M^{2n}$, then around every point $x$ in $M$ there is a local chart $(U_x ; q^1, \ldots, q^n, p_1, \ldots, p_n )$ and a function $f$ on $U_x$ such that
$$ \omega = e^f \sum_{i}^{} dq^i \wedge dp_i ~~ \text{ and } \theta = df = \sum_{i}^{} \big( \frac{\partial f}{\partial q^i}  dq^i  + \frac{\partial f}{\partial p_i}  dp_i \big).$$
Therefore, the induced Jacobi structure $(\pi, E)$ on $M$ is given by 
$$ \pi = e^{-f} \sum_{i}^{} \frac{\partial}{\partial q^i} \wedge \frac{\partial}{\partial p_i} ~~ \text{ and } ~~ E = e^{-f} \sum_{i}^{} \big(    \frac{\partial f}{\partial p_i} \frac{\partial}{\partial q^i} 
- \frac{\partial f}{\partial q^i} \frac{\partial}{\partial p_i} \big).$$
Hence, around the local chart $(q^1, \ldots, q^n, p_1, \ldots, p_n, t)$ on $M \times \mathbb{R}$, the Poissonization is given by the formula (\ref{pois-formula}) where $(\pi,E)$ is defined above. 
\end{remark}

A Jacobi structure $(\pi, E)$ on $M$ is called 'transitive' if the map
$$ pr_1 \circ (\pi, E)^\sharp : T^*M \times \mathbb{R} \rightarrow TM, ~ (\alpha, g) \mapsto \pi^\sharp \alpha + g E$$
is surjective. The only transitive Jacobi structures are given by
contact structures on odd dimensional manifolds and locally conformal symplectic (l.c.s.) structures on even dimensional manifolds \cite{kirillov}.
More generally, if $(\pi,E)$ is an arbitrary Jacobi structure, the leaves of its characteristic distribution $\mathcal{D}$ carries transitive Jacobi structures.

A conformal change of a Jacobi structure $(\pi, E)$ on $M$ by a nowhere vanishing smooth function $\sigma$ is given by $(\pi_\sigma , E_\sigma)$ where
$\pi_\sigma = \sigma \pi$ and $E_\sigma = \pi^\sharp (d \sigma) + \sigma E$.

\begin{defn}
A Jacobi map between two Jacobi manifolds $(\widetilde{M}, \widetilde{\pi}, \widetilde{E})$ and $(M, \pi,E)$ is a smooth map $\phi : \widetilde{M} \rightarrow M$
which preserves the corresponding bivector fields and vector fields, that is, $\phi_* \widetilde{\pi} = \pi$ and $\phi_* \widetilde{E} = E$. The map $\phi$ is
called a  conformal Jacobi map with respect to a nowhere vanishing function $\sigma \in C^\infty(\widetilde{M})$ if $\phi : \widetilde{M} \rightarrow M$ is a Jacobi map when $\widetilde{M}$ is
equipped with the conformal Jacobi structure $(\widetilde{\pi}_\sigma, \widetilde{E}_\sigma)$. We denote a conformal Jacobi map simply by $(\phi, \sigma) : \widetilde{M} \rightarrow M$.
\end{defn}
\medskip

Next we recall Dirac-Jacobi structure (or $\mathcal{E}^1(M)$-Dirac structure) on a manifold studied by Wade \cite{wade1}. A line bundle approach of this notion was further studied by Vitagliano \cite{vitagliano}. First observe that for any smooth manifold $M$,
the bundle $TM \times \mathbb{R} \rightarrow M$ has a Lie algebroid structure whose bracket and anchor are given by
$$ [(X, f) , (Y, h)] = ([X, Y], X(h) - Y(f)) ~~ \text{~ and ~}~~ \rho(X, f) = X,$$
for $(X, f), (Y, h) \in \Gamma (TM \times \mathbb{R}) \cong \Gamma (TM) \times C^\infty(M)$.
Moreover, for any $k \geq 0$, there is a square zero map $\widetilde{d} : \Omega^k(M) \times \Omega^{k-1}(M) \rightarrow \Omega^{k+1}(M) \times \Omega^{k}(M)$ defined by
$$ \widetilde{d} (\alpha, \beta) = (d \alpha, \alpha - d \beta), ~~~~ \text{ ~ for } (\alpha, \beta) \in \Omega^k(M) \times \Omega^{k-1}(M).$$
For any $(X,f) \in \Gamma (TM) \times C^\infty(M),$ there is also a contraction map $i_{(X,f)} : \Omega^k(M) \times \Omega^{k-1}(M) \rightarrow \Omega^{k-1}(M) \times \Omega^{k-2}(M)$
defined by
$$i_{(X,f)} (\alpha, \beta) = (i_X \alpha + f \beta , - i_X \beta ), ~~~~ \text{ ~ for } (\alpha, \beta) \in \Omega^k(M) \times \Omega^{k-1}(M).$$
Therefore, for any $(X,f) \in \Gamma (TM) \times C^\infty(M),$ one can define an operator $\widetilde{\mathcal{L}}_{(X,f)} :  \Omega^k(M) \times \Omega^{k-1}(M) \rightarrow  \Omega^k(M) \times \Omega^{k-1}(M)$
by the following Cartan like formula
$$ \widetilde{\mathcal{L}}_{(X,f)} := i_{(X,f)} \circ \widetilde{d} + \widetilde{d} \circ i_{(X,f)}.$$
Then one can verify that the following identity holds
$$\widetilde{\mathcal{L}}_{(X,f)} \circ i_{(Y,h)} - i_{(Y,h)} \circ \widetilde{\mathcal{L}}_{(X,f)} = i_{[(X,f),(Y,h)]} , ~~ \text{ for } (X,f) , (Y,h) \in \Gamma(TM) \times C^\infty(M).$$

Hence, we can define a non-degenerate pairing $ \langle\langle -,-\rangle\rangle$ on the space of sections of the bundle
$\mathcal{E}^1(M) := (TM \times \mathbb{R}) \oplus (T^*M \times \mathbb{R})$ by the following
\begin{align}\label{pairing}
 \langle\langle (X,f) \oplus (\alpha, g)  ,  (Y,h) \oplus (\beta , k) \rangle\rangle= \frac{1}{2}  (i_{(X,f)} (\beta, k)   + i_{(Y,h)} (\alpha, g))
\end{align}
and a generalized Dorfman bracket $ \llbracket-,-\rrbracket$ on the space of sections of $\mathcal{E}^1(M)$ by
$$ \llbracket(X,f) \oplus (\alpha, g) , (Y,h) \oplus (\beta , k)  \rrbracket =   ( [(X,f), (Y,h)]  \oplus \widetilde{\mathcal{L}}_{(X,f)} (\beta,k) - i_{(Y,h)} \widetilde{d} (\alpha,g) ) ,$$
for $(X,f) \oplus (\alpha, g) , (Y,h) \oplus (\beta , k)  \in \Gamma (\mathcal{E}^1(M)).$

\begin{defn}\label{dj-defn}
 A Dirac-Jacobi structure on a manifold $M$ is a subbundle
$$L \subset (TM \times \mathbb{R}) \oplus (T^*M \times \mathbb{R})$$
which is maximally isotropic with respect to the pairing $\langle\langle - , - \rangle\rangle$ and such that $\Gamma L$ is closed under the generalized Dorfman bracket $\llbracket -, - \rrbracket$.
\end{defn}

The kernel of a Dirac-Jacobi structure $L$ is defined as $ker ~(L) := L \cap ((TM \times \mathbb{R}) \oplus \{ 0 \}).$
Given a Dirac-Jacobi structure $L$, its opposite Dirac-Jacobi structure is given by
$$ L_{-} := \{ (- X, - f) \oplus (\alpha, g) ~|~ ( X, f) \oplus (\alpha, g) \in L \}.$$

\begin{remark}
Note that if $L$ is a Dirac-Jacobi subbundle then it is equipped with a Lie algebroid structure over $M$. The Lie bracket on $\Gamma L$ is given by the restriction of the generalized 
Dorfman bracket $\llbracket - , - \rrbracket$ and the anchor is given by the projection on $TM$. Moreover, there is a distinguised Lie algebroid
$1$-cocycle given by 
$$ (X, f) \oplus (\alpha, g) \mapsto f, ~~ \text{ for } (X, f) \oplus (\alpha, g) \in \Gamma L.$$
\end{remark}

Given a Dirac-Jacobi structure $L \subset (TM \times \mathbb{R}) \oplus (T^*M \times \mathbb{R})$ on $M$, one can correspond a Dirac structure $\widetilde{L}$
on $M \times \mathbb{R}$ given by following Diracization \cite{igl-wade} process
$$ \widetilde{L} = \{  (X + f \frac{\partial}{\partial t}  ) \oplus e^t (\alpha + g dt) |~ (X,f) \oplus (\alpha , g) \in L    \}    .$$

A characterization of Jacobi structures is given by the following \cite{wade1, nunes-gallardo}.
\begin{prop}\label{jacobi-graph}
 Let $M$ be a smooth manifold and $(\pi, E)$ be a pair of a bivector field and a vector field on $M$. Then $(\pi, E)$ defines a Jacobi
structure on $M$ if and only if
$$ L_{(\pi,E)} := Graph ((\pi,E)^\sharp) = \{ (\pi^\sharp \alpha + g E, - \langle \alpha, E \rangle) \oplus (\alpha , g) |~ (\alpha, g) \in T^*M \times \mathbb{R} \}$$
is a Dirac-Jacobi structure on $M$.
\end{prop}

\begin{remark}
 Let $(\pi, E)$ be a Jacobi structure on $M$ with the corresponding Dirac-Jacobi structure 
$ L_{(\pi, E)}$.
By identifying $T^*M \times \mathbb{R}$ with $ L_{(\pi, E)}$, we get a Lie algebroid structure on $T^*M \times \mathbb{R}$ with bracket
$$ [(\alpha, g), (\beta,k) ]_{(\pi,E)}  = \widetilde{\mathcal{L}}_{(\pi,E)^\sharp (\alpha, g)} (\beta,k)  - i_{(\pi,E)^\sharp (\beta, k)} \widetilde{d} (\alpha, g), ~~ \text{ for } (\alpha, g), (\beta, k) \in \Gamma (T^*M \times \mathbb{R}),$$
and the anchor $\rho : T^*M \times \mathbb{R} \rightarrow TM$ is given by $\rho := pr_1 \circ (\pi,E)^\sharp$. Moreover, the distinguised $1$-cocycle of this Lie algebroid is given by
$(\alpha, g) \mapsto - \langle \alpha , E \rangle$. Hence, it is given by $(-E, 0) \in \Gamma (TM \times \mathbb{R}).$
With this Lie algebroid structure on $T^*M \times \mathbb{R}$, the bundle
map $(\pi,E)^\sharp : T^*M \times \mathbb{R} \rightarrow TM \times \mathbb{R}$ defined in (\ref{sharp-map}) is a Lie algebroid morphism. 

Moreover, we remark that the corresponding Diracization $\widetilde{L_{(\pi, E)}}$ on $M \times \mathbb{R}$ is given by the graph of the Poissonization.
In other words, $\widetilde{L_{(\pi, E)}} = L_{\widetilde{(\pi, E)}}.$ 
\end{remark}
\begin{remark}\label{dj-dirac-p}
Let $L$ be a Dirac-Jacobi structure on $M$ such that its Diracization $\widetilde{L}$ is given by a Poisson structure on $M \times \mathbb{R}$. Then it follows from the definition of $\widetilde{L}$ that it must be of the following form
$$ \widetilde{L} = L_{e^{-t} (\pi + \frac{\partial}{\partial t} \wedge E)}$$
for some bivector field $\pi$ on $M$ and a vector field $E$ on $M$. The Dirac-Jacobi structure $L$ is then given by $L = L_{(\pi,E)}.$ Moreover, by Proposition \ref{jacobi-graph} defines a Jacobi structure on $M$.
\end{remark}

\medskip

Another example of Dirac-Jacobi structure is given by a precontact $1$-form. The term 'precontact' is just a terminology and this is nothing but a usual
$1$-form.
Given a $1$-form $\eta$ on $M$, one can define a Dirac-Jacobi structure on $M$ given by
\begin{align}\label{pre-dj}
L_\eta = \big\{  (X,f) \oplus (i_X d \eta + f \eta, - i_X \eta) | ~(X,f) \in TM \times \mathbb{R}   \big\} \subset (TM \times \mathbb{R}) \oplus (T^*M \times \mathbb{R}).
\end{align}

The next theorem suggests when a Dirac-Jacobi structure $L$ comes from a contact $1$-form $\eta$ \cite{wade2}.
\begin{thm}\label{contact-dirac-jacobi}
 A Dirac-Jacobi structure $L$ on an odd dimensinal manifold $M$ is given by a contact $1$-form $\eta$ if and only if it satisfies
\begin{align*}
L_x \cap (   (T_xM \times \mathbb{R}) \oplus \{0\}) =~& \{0\}, \\
L_x \cap (   \{0\} \oplus (T_x^*M \times \mathbb{R})) =~& \{0\}, 
\end{align*}
for all $ x \in M.$
\end{thm}

\medskip

\section{Gauge transformations}\label{sec-3}
In this section, we introduce gauge transformations of Dirac-Jacobi and Jacobi structures on a manifold.

Let $L \subset (TM \times \mathbb{R}) \oplus (T^*M \times \mathbb{R})$ be a Dirac-Jacobi structure on $M$ and take $(B_1, B) \in \Omega^2(M) \times \Omega^1(M)$ be a pair of a $2$-form and a $1$-form.
Consider the subbundle
\begin{align*}
\tau_{(B_1, B)} (L) :=&~ \{ (X,f) \oplus (\alpha, g) + i_{(X,f)} {(B_1, B)} ~|~ (X,f) \oplus (\alpha, g) \in L \}\\
= &~ \{ (X,f) \oplus (\alpha + i_X B_1 + f B~,~ g - \langle X, B \rangle) ~|~ (X,f) \oplus (\alpha, g) \in L \}.
\end{align*}
It is easy to see that the bundle $\tau_{(B_1, B)} (L) \subset (TM \times \mathbb{R}) \oplus (T^*M \times \mathbb{R})$ is also maximally isotropic with respect to the pairing $\langle \langle ~, ~ \rangle \rangle.$
Moreover, we have the following.

\begin{lemma}
For any $1$-form $B$, the space of sections of the bundle
 $\tau_{(d B,B)} (L)$ is closed under the generalized Dorfman bracket.
\end{lemma}
\begin{proof}
For any $(X,f) \oplus (\alpha, g)  ~,~  (Y,h) \oplus (\beta, k) \in \Gamma L,$ we have
\begin{align*}
& \llbracket (X,f) \oplus (\alpha, g) + i_{(X,f)} (d B,B)~ ,~ (Y,h) \oplus (\beta, k) + i_{(Y,h)} (d B,B) \rrbracket \\
=&~ [(X,f), (Y,h)] \oplus ( \widetilde{\mathcal{L}}_{(X,f)} (\beta,k) + \widetilde{\mathcal{L}}_{(X,f)} i_{(Y,h)} (d B,B)
- i_{(Y,h)} \widetilde{d} (\alpha,g) - i_{(Y,h)} \widetilde{d} i_{(X,f)} (d B,B) )\\
=&~ [(X,f), (Y,h)] \oplus ( \widetilde{\mathcal{L}}_{(X,f)} (\beta,k)  -   i_{(Y,h)} \widetilde{d} (\alpha,g)        + \widetilde{\mathcal{L}}_{(X,f)} i_{(Y,h)} (d B,B)-   i_{(Y,h)}   \widetilde{\mathcal{L}}_{(X,f)} (d B, B) \\
 & + i_{(Y,h)} i_{(X,f)} \widetilde{d}  (d B,B))  \\
=&~ [(X,f), (Y,h)] \oplus ( \widetilde{\mathcal{L}}_{(X,f)} (\beta,k)  -   i_{(Y,h)} \widetilde{d} (\alpha,g)   + i_{[(X,f),(Y,h)]} (d B,B)   + i_{(Y,h)} i_{(X,f)} \widetilde{d}  (d B,B) ).
\end{align*}
Since $\widetilde{d} (d B , B) = 0$, it follows that $\Gamma (\tau_{(d B, B)} (L))$ is closed under the generalized Dorfman bracket. 
\end{proof}

Therefore, for any $1$-form $B \in \Omega^1(M)$, the transformation $\tau_{(d B, B)} (L)$ defines a Dirac-Jacobi structure on $M$.
We denote this Dirac-Jacobi structure simply by $\tau_B (L)$ and is called the gauge transformation of $L$ associated
to the $1$-form $B$. The Dirac-Jacobi structure $L$ and $\tau_B (L)$ are called gauge equivalent. We remark that when $L = L_\eta$, the gauge transformation $\tau_B(L)$ is given by
$L_{\eta + B}.$

The proof of the following is obvious.
\begin{prop}\label{gauge-alg-iso}
 Gauge transformations of Dirac-Jacobi structures satisfy the following properties.
\begin{itemize}
 \item[(i)] $\tau_0 (L) = L$ and $\tau_{B'} (\tau_B (L)) = \tau_{B} (\tau_{B'} (L)) = \tau_{B + B'} (L)$, hence, gauge transformations defines an action of the abelian group
$\Omega^1(M)$ on the space $\mathbb{DJ}(M)$ of all Dirac-Jacobi structures on $M$.
 \item[(ii)] The map $\tau_B : L \rightarrow \tau_B(L)$ defined by $(X,f) \oplus (\alpha, g) \mapsto (X,f) \oplus (\alpha, g) + i_{(X,f)} (d B , B)$
defines an isomorphism between Lie algebroid structures and under this isomorphism, the distinguised $1$-cocycles are same.
\end{itemize}
\end{prop}

\medskip

Next, we consider Dirac-Jacobi structures on $M$ which are graph of Jacobi structures.
Let $(\pi, E)$ be a Jacobi structure on $M$. Take any $B \in \Omega^1(M)$. Consider
the Dirac-Jacobi structure $\tau_B (L_{(\pi, E)})$ gauge equivalent to $L_{(\pi, E)}$,
$$ \tau_B (L_{(\pi, E)}) = \{ (\pi^\sharp \alpha + g E , - \langle \alpha, E \rangle) \oplus ( \alpha, g) + i_{(\pi,E)^\sharp (\alpha,g)} (d B, B)) ~|~ (\alpha, g) \in T^*M \times \mathbb{R} \} .$$
Let $\widetilde{(d B, B)} : TM \times \mathbb{R} \rightarrow T^*M \times \mathbb{R},~ (X,f) \mapsto i_{(X,f)} (d B, B)$ be the bundle map induced by $(d B, B)$. If the bundle map
\begin{align}\label{gauge-invertible}
\big( Id + \widetilde{(d B, B)} \circ (\pi, E)^\sharp \big) : T^*M \times \mathbb{R} \longrightarrow T^*M \times \mathbb{R}
\end{align}
is invertible, then $\tau_B (L_{(\pi,E)})$ is the graph of the map 
\begin{align}\label{skew-gauge-map}
(\pi,E)^\sharp \big( Id + \widetilde{(d B, B)} \circ (\pi, E)^\sharp \big)^{-1} : T^*M \times \mathbb{R} \longrightarrow TM \times \mathbb{R}.
\end{align}
In this case, the $1$-form $B$ is called $(\pi,E)$-admissible.
Moreover, the
map defined in (\ref{skew-gauge-map}) is skew-symmetric, thus, given by a pair $(\pi_B, E_B)$ of a bivector field $\pi_B$ and a vector field $E_B$ on $M$. The pair $(\pi_B, E_B)$ is completely
determined by
$$ (\pi_B, E_B)^\sharp = (\pi, E)^\sharp (Id + \widetilde{(d B, B)} \circ (\pi, E)^\sharp)^{-1},$$
and, in this case,
$$ \tau_B (L_{(\pi,E)}) = Graph ~( (\pi_B, E_B)^\sharp ) = L_{(\pi_B, E_B)}.$$
Therefore, it follows from Proposition \ref{jacobi-graph} that $(\pi_B, E_B)$ is a Jacobi structure on $M$. 
The Jacobi structure $(\pi_B, E_B)$ which is also denoted by $\tau_B (\pi,E)$, is called the gauge transformation of $(\pi, E)$ associated with the $1$-form $B$.
The Jacobi structures $(\pi,E) ,~ (\pi_B, E_B)$ are called gauge equivalent.

\begin{remark}\label{gauge-n-p-same-fol}
 Since the map (\ref{gauge-invertible}) is an isomorphism, it follows that Im $(pr_1 \circ (\pi,E)^\sharp)$ = Im $(pr_1 \circ (\pi_B, E_B)^\sharp)$. 
Therefore, gauge equivalent Jacobi structures gives rise to same characteristic distribution.
\end{remark}

More generally, gauge equivalent Jacobi structures on $M$ correspond to isomorphic Lie algebroids on $T^*M \times \mathbb{R}$.
\begin{prop}\label{gauge-iso-lie}
 Let $(\pi,E)$ be a Jacobi structure on $M$, and $(\pi_B, E_B)$ be a gauge equivalent Jacobi structure associated with the $1$-form $B$.
Then the Lie algebroid structures on $T^*M \times \mathbb{R}$ associated to $(\pi,E)$ and $ (\pi_B, E_B)$ are isomorphic.
\end{prop}
\begin{proof}
Consider the bundle isomorphism $\Phi := (Id ~ + ~ \widetilde{(d B, B)} \circ (\pi,E)^\sharp): T^*M \times \mathbb{R} \rightarrow T^*M \times \mathbb{R}$, given by  $(\alpha , g) \mapsto (\alpha , g) + i_{(\pi, E)^\sharp (\alpha,g)} (d B, B)$,
for $(\alpha , g) \in T^*M \times \mathbb{R}$. This map commute with the corresponding anchors, as
$$ (\pi,E)^\sharp =  (\pi_B, E_B)^\sharp \circ (Id~ + ~ \widetilde{(d B, B)} \circ (\pi,E)^\sharp) = (\pi_B, E_B)^\sharp \circ \Phi.$$
For any $(\alpha, g), (\beta, k) \in \Gamma (T^*M \times \mathbb{R})$, we also have
\begin{align*}
& [ \Phi (\alpha, g) , \Phi (\beta ,  k) ]_{(\pi_B, E_B)}\\
=&~ \widetilde{\mathcal{L}}_{ (\pi_B,E_B)^\sharp \Phi (\alpha, g)} \Phi (\beta , k ) - i_{ (\pi_B, E_B)^\sharp \Phi (\beta , k)} \widetilde{d} (\Phi (\alpha,g))\\
=&~ \widetilde{\mathcal{L}}_{ (\pi,E)^\sharp (\alpha, g)} (\beta, k) + \widetilde{\mathcal{L}}_{ (\pi,E)^\sharp (\alpha, g)} i_{(\pi,E)^\sharp (\beta,k)} (d B, B) - i_{(\pi,E)^\sharp (\beta. k)} \widetilde{d} (\alpha, g) - i_{(\pi,E)^\sharp (\beta. k)} \widetilde{d} i_{(\pi,E)^\sharp (\alpha, g)} (d B, B)\\
=&~ [ (\alpha, g) , (\beta , k) ]_{(\pi,E)} + i_{[(\pi,E)^\sharp (\alpha, g), (\pi,E)^\sharp (\beta, k) ]} (d B, B) = \Phi ([ (\alpha, g), (\beta , k ) ]_{(\pi,E)} ).
\end{align*}
Hence the proof.
\end{proof}

Note that the isomorphism $\Phi$ pulls back the cocycle $(-E_B, 0)$ to the cocycle $(-E, 0)$. This can be shown by considering the distinguised cocycles
of the corresponding Dirac-Jacobi structures.

\begin{remark}
Gauge equivalent Jacobi structures on $M$ gives rise to isomorphic Lie algebroid cohomology of $T^*M \times \mathbb{R}$. In other words, they have isomorphic Lichnerowicz-Jacobi cohomology.
\end{remark}

\begin{remark}\label{lcs-contact-remark}
 Let $(\pi, E)$ be a transitive Jacobi structure on $M$. Let $B  \in \Omega^1 (M)$ be such that the gauge transformation
$\tau_B (L_{(\pi,E)})$ defines a Jacobi structure $(\pi_B, E_B)$ on $M$. Then we have
$$ (\pi_B, E_B)^\sharp = (\pi, E)^\sharp (Id + \widetilde{(d B, B)} \circ (\pi, E)^\sharp)^{-1}.$$
Therefore,
$$ pr_1 \circ  (\pi_B, E_B)^\sharp = pr_1 \circ (\pi,E)^\sharp \circ (Id + \widetilde{(d B, B)} \circ (\pi, E)^\sharp)^{-1}.$$
Since the map $\big( Id + \widetilde{(d B, B)} \circ (\pi, E)^\sharp \big)$ is invertible and $ pr_1 \circ (\pi, E)^\sharp : T^*M \times \mathbb{R} \rightarrow TM$
is surjective, it follows that $ pr_1 \circ (\pi_B, E_B)^\sharp $ is also surjective. Therefore, gauge transformations of transitive Jacobi structures are transitive.
Thus, gauge transformations of contact structures are contact and gauge transformations of l.c.s. structures are l.c.s. In the next remark we give an alternative
argument of this fact for contact structures.
Since a gauge transformation of a Jacobi structure preserves the characteristic distribution, it only transform the contact or l.c.s. structures on the characteristic leaves by 
the pullback of the $1$-form $B$.
\end{remark}

\begin{remark}\label{contact-remark-1}
 Let $\eta$ be a contact $1$-form on $M$ with associated Jacobi structure $ (\pi, E)$. Let $(\pi_B, E_B)$ be its gauge transformation associated with the $(\pi,E)$-admissible $1$-form
$B$. As we have
$$ (\pi_B, E_B)^\sharp = (\pi,E)^\sharp  (Id + \widetilde{(d B, B)} \circ (\pi, E)^\sharp)^{-1},$$
the map $(\pi_B, E_B)^\sharp$ is also invertible and the inverse is given by
\begin{align}\label{gauge-contact-inverse}
( (\pi_E, E_B)^\sharp )^{-1} = ((\pi,E)^\sharp)^{-1} + \widetilde{(d B, B)}.
\end{align}
Therefore, the gauge transformation of contact structures are contact.

If $\eta_B$ denotes the gauge transformation of $\eta$ associated with $B$, then it follows from (\ref{contact-inverse}) and (\ref{gauge-contact-inverse}) that
$$ (\eta_B, 0) = ((\pi_B, E_B)^\sharp )^{-1} (0, -1) = ((\pi,E)^\sharp )^{-1} (0, -1) + ~ i_{(0, -1)} {(d B,B)} = (\eta, 0) ~+~ (- B , 0) = (\eta - B , 0).$$
Therefore, we have $\eta_B = \eta - B.$
\end{remark}

\begin{remark}\label{contact-remark-2}
Let $\eta$ and $\eta'$ be any two contact structures on $M$ with associated Jacobi structures $ (\pi, E)$ and $ (\pi', E')$, respectively. 
Define a map 
$\theta : TM \times \mathbb{R} \rightarrow T^*M \times \mathbb{R}$ by
$$ ((\pi', E')^\sharp)^{-1} = ((\pi,E)^\sharp)^{-1} + \theta.$$
It follows from Equation (\ref{contact-inverse}) that the bundle map $((\pi,E)^\sharp)^{-1} : TM \times \mathbb{R} \rightarrow T^*M \times \mathbb{R}$ is given by
$\widetilde{(-d \eta, - \eta)}$. Similarly, for the bundle map $ ((\pi', E')^\sharp)^{-1}$. This shows that the map
$\theta :  TM \times \mathbb{R} \rightarrow T^*M \times \mathbb{R}$ is skew-symmetric and is given by
$\widetilde{(d B, B)}$, where $B = \eta - \eta'.$
It also follows from the construction that the map
$\big( Id + \widetilde{(d B,B)} \circ (\pi,E)^\sharp \big)$ is invertible. This shows that any two contact structures on $M$ are gauge equivalent.
\end{remark}

Next, we discuss the effect of gauge transformations on l.c.s. structures. Let $(\omega, \theta)$ be a l.c.s. structure on $M$ with associated Jacobi structure
$(\pi,E)$. See Example \ref{exam-lcs} for details. The corresponding Dirac-Jacobi structure is given by
\begin{align}\label{lcs-eqn}
 L_{(\pi,E)} =&~ \{ (\pi^\sharp \alpha + g E, - \langle \alpha, E \rangle) \oplus (\alpha, g) ~|~ ( \alpha , g ) \in T^*M \times \mathbb{R} \} \nonumber \\ 
=&~ \{ (- \flat^{-1} (\alpha ) + g~ \flat^{-1} (\theta), - \langle \alpha, \flat^{-1}(\theta) \rangle~) \oplus (\alpha, g) ~|~ (\alpha, g) \in T^*M \times \mathbb{R} \} \nonumber \\
=&~ \{ ( - \flat^{-1} (\alpha - g \theta), ~ - \langle \alpha, \flat^{-1} (\theta) \rangle ~) \oplus (\alpha , g)  ~|~ (\alpha, g) \in T^*M \times \mathbb{R} \}\nonumber \\
=&~ \{ ( - \flat^{-1} (\alpha), - \langle \alpha + g \theta , \flat^{-1} (\theta) \rangle ~) \oplus (\alpha + g \theta, g)  ~|~ (\alpha, g) \in T^*M \times \mathbb{R} \}\nonumber \\
=&~ \{ ( - X,  \theta (X) ) \oplus (i_X \omega + g \theta, g) ~|~ (X, g) \in TM \times \mathbb{R} \}. 
\end{align}
For any $1$-form $B \in \Omega^1(M)$, we have from Equation (\ref{lcs-eqn}) that
\begin{align}\label{lcs-eqn-2}
&\big(  Id + \widetilde{(d B, B)} \circ (\pi,E)^\sharp \big) \big(i_X \omega  - B(X) \theta,~ - B(X) \big)  \nonumber \\
=&~ (i_X \omega  - B(X) \theta,~ - B(X))  + i_{( - X,  \theta (X) )} (d B, B) \nonumber \\
=&~ (i_X \omega - B(X) \theta,~ - B(X)) + (- i_X d B + \theta(X) B,~ B (X)) \nonumber \\
=&~ (i_X \omega - i_X d B - B(X) \theta + \theta (X) B  ,~ 0 ) \nonumber \\
=&~ (i_X (\omega - d B - B \wedge \theta) ,~ 0).
\end{align}
If $B$ is $(\pi,E)$-admissible, that is, the map $( Id + \widetilde{(d B, B)} \circ (\pi,E)^\sharp )$ is invertible,
then it follows from (\ref{lcs-eqn-2}) that the two form $(\omega - d B - B \wedge \theta)$ is non-degenerate. In this case, the pair $(\omega - d B - B \wedge \theta, \theta)$ defines
a l.c.s. structure on $M$. Moreover,
\begin{align*}
 \tau_B (L_{(\pi,E)}) =&~ \{  ( - X,  \theta (X) ) \oplus (i_X \omega + g \theta, g) + i_{( - X,  \theta (X) )} (d B , B) ~|~ (X, g) \in TM \times \mathbb{R} \}\\
=&~ \{  ( - X,  \theta (X) ) \oplus  (i_X \omega + g \theta - i_X d B + \theta (X) B ,~ g + B (X) ) ~|~  (X, g) \in TM \times \mathbb{R} \}\\
=&~  \{  ( - X,  \theta (X) ) \oplus \big(i_X (\omega - d B - B \wedge \theta) + g \theta ,~ g  \big) ~|~  (X, g) \in TM \times \mathbb{R} \}
\end{align*}
is the Dirac-Jacobi structure induced from the l.c.s. structure $(\omega - d B - B \wedge \theta, \theta)$. Hence, the gauge transformation of the l.c.s.
structure $(\omega, \theta)$ is given by $(\omega - d B - B \wedge \theta, \theta)$.

\medskip

\section{Gauge transformations commute with the Poissonization}\label{sec-poissonization}
In this section, we prove that gauge transformations of Jacobi structures on a manifold commute with the Poissonization process.
We first observe that gauge transformations of Dirac-Jacobi structures commute with the Diracization.

Let $L \subset (TM \times \mathbb{R}) \oplus (T^*M \times \mathbb{R})$ be a Dirac-Jacobi structure on $M$. Take $B \in \Omega^1(M).$
Then we have the following.

\begin{prop}\label{gauge-commute-diracization}
 Gauge transformations of Dirac-Jacobi structures commute with the Diracization process. In other words, the following diagram commute
\begin{align*}
\xymatrixrowsep{0.5in}
\xymatrixcolsep{0.7in}
\xymatrix{
(M, L) \ar[r]^{\tau_B} \ar[d]_{Dirac} & (M, \tau_B (L)) \ar[d]^{Dirac} \\
(M \times \mathbb{R} , \widetilde{L}) \ar[r]_{\tau_{\widetilde{B}}} & (M \times \mathbb{R}, \tau_{\widetilde{B}} (\widetilde{L}) = \widetilde{\tau_B (L)}))
}
\end{align*}
where $\widetilde{B} = e^t (pr_1^* d B + d t \wedge pr_1^* B)$ is a closed $2$-form on $M \times \mathbb{R}$.
\end{prop}

\begin{proof}
We have
\begin{align*}
 \tau_B (L) =~& \{   (X,f) \oplus ((\alpha, g) + i_{(X,f)} (d B, B) ) ~|~ (X,f) \oplus (\alpha, g) \in L \} \\
=~& \{   (X,f) \oplus ((\alpha, g) + (i_X d B + f B , - \langle B, X \rangle) ) ~ |~ (X,f) \oplus (\alpha, g) \in L  \} \\
=~& \{   (X,f) \oplus ( \alpha +  i_X d B + f B , g - \langle B , X \rangle    ) ~ |~ (X,f) \oplus (\alpha, g) \in L  \}.
\end{align*}
Hence,
$$ \widetilde{\tau_B (L)} = \{  (X + f \frac{\partial}{\partial t} ) \oplus e^t( \alpha +  i_X d B + f B + g dt - \langle B, X \rangle dt)  ~ |~ (X,f) \oplus (\alpha, g) \in L    \}.$$
On the other hand, from the definition of $\widetilde{L}$ it follows that
\begin{align*}
 \tau_{\widetilde{B}} (\widetilde{L}) =~& \{    (X + f \frac{\partial}{\partial t} ) \oplus e^t (\alpha + g dt) + i_{X + f \frac{\partial}{\partial t}} e^t (pr_1^* d B + d t \wedge pr_1^* B)  ~|~ (X,f) \oplus (\alpha, g) \in L \} \\
=~& \{ (X + f \frac{\partial}{\partial t} ) \oplus e^t( \alpha + g dt + i_X d B  - \langle B , X \rangle dt + f B )  ~ |~ (X,f) \oplus (\alpha, g) \in L    \} \\
=~& \widetilde{\tau_B (L)}.
\end{align*}
\end{proof}

To prove that gauge transformations of Jacobi structures commute with the Poissonization, we need the following lemma.

\begin{lemma}\label{adm-adm}
 Let $(\pi,E)$ be a Jacobi structure on $M$ and $B \in \Omega^1(M)$ be a $1$-form. Then $B$ is $(\pi,E)$-admissible if and only if $\widetilde{B}$ is $\widetilde{(\pi,E)}$-admissible.
\end{lemma}
\begin{proof}
 For any $1$-form $B$, we have from Proposition \ref{gauge-commute-diracization} that
\begin{align}\label{adm-adm}
 \tau_{\widetilde{B}} (L_{ \widetilde{(\pi,E)}  })=  \tau_{\widetilde{B}} ( \widetilde{L_{(\pi,E)}} ) = \widetilde{ \tau_B (L_{(\pi,E)})}.
\end{align}
If $B$ is $(\pi, E)$-admissible, then $\widetilde{ \tau_B (L_{(\pi,E)})}  = \widetilde{ L_{\tau_B (\pi,E)  }  } = L_{\widetilde{ \tau_B (\pi,E)   }}.$ Hence, we have
$$ \tau_{\widetilde{B}} (L_{ \widetilde{(\pi,E)}  }) =   L_{\widetilde{ \tau_B (\pi,E)   }}.$$
Since $\widetilde{(\pi,E)}$ and $\widetilde{ \tau_B (\pi,E)   }$ are both Poisson structures on $M \times \mathbb{R}$, the above equality holds only when $\widetilde{B}$
is $\widetilde{(\pi,E)}$-admissible and $\tau_{\widetilde{B}} \widetilde{(\pi,E)} =  \widetilde{ \tau_B (\pi,E)   }$.

Conversely, if $\widetilde{B}$ is $\widetilde{(\pi,E)}$-admissible, then $\tau_{\widetilde{B}} (L_{\widetilde{(\pi,E)}} )= L_{\tau_{\widetilde{B}} \widetilde{(\pi,E)}}$. Therefore, we have from (\ref{adm-adm}) that
$$ L_{\tau_{\widetilde{B}} \widetilde{(\pi,E)}} = \widetilde{\tau_B (L_{(\pi,E)})}.$$
Since $\tau_{\widetilde{B}} \widetilde{(\pi,E)}$ defines a Poisson structure on $M \times \mathbb{R}$, it follows from  Remark \ref{dj-dirac-p} that the Dirac-Jacobi structure $\tau_B (L_{(\pi,E)})$ is given by a Jacobi structure. In other words, $B$ is $(\pi,E)$-admissible.
\end{proof}

By Proposition \ref{gauge-commute-diracization} and Lemma \ref{adm-adm} we have the following.

\begin{prop}\label{gauge-commute-poissonization}
Let $(M, \pi, E)$ be a Jacobi manifold 
and $B \in \Omega^1(M)$ be a $(\pi,E)$-admissible $1$-form on $M$. 
Then $\tau_{\widetilde{B}} \widetilde{(\pi,E)} = \widetilde{\tau_B (\pi, E)}$, where $\widetilde{(\pi,E)}$ and $\widetilde{\tau_B (\pi,E)}$ denote the Poissonization of the Jacobi structures $(\pi,E)$ and $\tau_B (\pi,E)$, respectively.

\begin{align*}
\xymatrixrowsep{0.5in}
\xymatrixcolsep{0.7in}
\xymatrix{
(M, \pi,E) \ar[r]^{\tau_B} \ar[d]_{Pois} & (M, \tau_B (\pi,E)) \ar[d]^{Pois} \\
(M \times \mathbb{R} , \widetilde{(\pi,E)}) \ar[r]_{\tau_{\widetilde{B}}} & (M \times \mathbb{R}, \tau_{\widetilde{B}} \widetilde{(\pi,E)} = \widetilde{\tau_B (\pi,E)})
}
\end{align*}
Therefore, gauge transformations of Jacobi structures commute with the Poissonization.
\end{prop}

\begin{remark}
In the particular case of contact structure, the proof is more simple. This follows from the  fact that the Poissonization of a contact structure $\eta$ on $M$
is given by a symplectic structure $\widetilde{\eta} = e^t (pr_1^* d \eta + d t \wedge  pr_1^* \eta)$ on $M \times \mathbb{R}$.

Hence, in this case,
\begin{align*}
 (\tau_{\widetilde{B}} \circ Pois )(\eta) =~&  \tau_{\widetilde{B}} \big(   e^t (pr_1^* d \eta + d t \wedge  pr_1^* \eta)    \big) \\
=~& e^t (pr_1^* d \eta + d t \wedge  pr_1^* \eta) - \widetilde{B} \\
=~& e^t (pr_1^* d \eta + d t \wedge  pr_1^* \eta - pr_1^* d B - d t \wedge pr_1^* B ) \\
=~& e^t \big( pr_1^* d (\eta - B) + d t \wedge pr_1^* (\eta - B) \big)\\
=~& Pois ~(\eta - B) = (Pois \circ \tau_B )(\eta).
\end{align*}
\end{remark}

\medskip

\section{Gauge transformations and contact groupoids}\label{sec-contact-groupoid}
The notion of contact groupoid was first introduced in \cite{ker-sou}. See \cite{cra-zhu, cra-sal} for recent developments on contact groupoids and integrable Jacobi structures.
In this section, we describe how the contact structure of a contact groupoid is effected by a gauge transformation of the Jacobi structure on its base.
\begin{defn}
 A contact groupoid is a Lie groupoid $G \rightrightarrows M$ together with a contact
$1$-form $\eta \in \Omega^1(G)$ and a function $\sigma \in C^\infty (G)$ such that
%\begin{itemize}
% \item[(i)] 
$$\eta_{gh} (X_g \oplus_{TG} Y_h) = \eta_g (X_g) + e^{\sigma (g)} \eta_h (Y_h), ~ \text{ for } (X_g, Y_h) \in (TG)^{(2)},$$
% \item[(ii)] $ker (d \eta)_x \cap ker (\eta)_x \cap ker (d \alpha)_x \cap ker (d \beta)_x \cap ker (d \sigma)_x = \{ 0 \}$, ~~ for all $x \in M$,
%\end{itemize}
where $\oplus_{TG}$ denotes the groupoid (partial) multiplication on the tangent Lie groupoid $TG \rightrightarrows TM$.
\end{defn}
It follows from the above condition that $\sigma$ is a multiplicative function on $G$, that is, $\sigma (gh) = \sigma (g) + \sigma (h),$ for all $(g,h) \in G^{(2)}.$
Given a contact groupoid $(G \rightrightarrows M, \eta, \sigma)$, the manifold $M$ carries a unique Jacobi structure such that $(\alpha, e^\sigma)$ is a conformal
Jacobi map and $\beta$ is an anti-Jacobi map (see \cite{cra-zhu, igl-marr2} for more details). 
%A Jacobi structure $(\pi, E)$ on $M$ is said to be integrable if there exists a contact groupoid over it. 
In this case, the contact groupoid $(G \rightrightarrows M, \eta, \sigma)$ is said to integrate the base Jacobi structure.
Alternatively, a Jacobi structure $(\pi,E)$ on $M$ is integrable if and only if the corresponding Lie algebroid structure on $T^*M \times \mathbb{R} \rightarrow M$ is integrable \cite{cra-zhu}.
In this case, the source-connected, simply connected Lie groupoid integrating $T^*M \times \mathbb{R} \rightarrow M$ carries a unique contact form and a multiplicative
function which makes it a contact groupoid. Under this correspondence, the multiplicative function on the groupoid differentiates to the distinguised $1$-cocycle $(-E,0)$ of the Lie algebroid.

Let $(M, \pi, E)$ be an integrable Jacobi manifold, with contact groupoid $(G \rightrightarrows M, \eta, \sigma)$. 
Let $B \in \Omega^1(M)$ be a $(\pi,E)$-admissible $1$-form on $M$. Since the Jacobi structures $(\pi,E)$ and $\tau_B (\pi,E)$
corresponds to isomorphic Lie algebroid structures on $T^*M \times \mathbb{R}$ and also the distinguised $1$-cocycles are same,
the Jacobi structure $\tau_B(\pi,E)$ and the corresponding distinguised $1$-cocycle can be integrated by a Lie groupoid isomorphic to $G \rightrightarrows M$ and by the multiplicative function $\sigma$ (Proposition \ref{gauge-iso-lie}).
Here we discuss the effect of the gauge transformation $\tau_B$ on the base Jacobi manifold to the contact $1$-form
of the Lie groupoid $G$ (Theorem \ref{final-thm}).

We first observe that if 
$\eta$ is a contact $1$-form with corresponding Jacobi structure $(\pi_\eta, E_\eta)$, then the corresponding Dirac-Jacobi structures are related by $L_{(\pi_\eta, E_\eta)} = (L_\eta)_{-} = L_{- \eta}.$

We recall the following definition from \cite{igl-wade}.

\begin{defn}
Let $\widetilde{M}$ (resp. $M$) be a smooth manifold and
 $L_{\widetilde{M}}$ (resp. $L_M$) a Dirac-Jacobi structure on $\widetilde{M}$ (resp. $M$). A smooth map $\phi : \widetilde{M} \rightarrow M$ is said to be a (forward) Dirac-Jacobi map
if $L_M = \phi_* (L_{\widetilde{M}})$, where
$$ \phi_* (L_{\widetilde{M}}) = \big\{  ( \phi_* \widetilde{X}, f) \oplus (\alpha, g)    |~  (\widetilde{X}, f \circ \phi) \oplus (\phi^* \alpha, g \circ \phi ) \in L_{\widetilde{M}} \big\}.$$
The map $\phi : \widetilde{M} \rightarrow M$ is called anti-Dirac-Jacobi map if $\phi_* (L_{\widetilde{M}}) = (L_M)_{-}.$ In any case,
 if $(\widetilde{X}, f \circ \phi) \in ker ~{L_{\widetilde{M}}}$ then $( \phi_* (\widetilde{X}), f) \in ker {L_M}.$
\end{defn}

We show that a smooth Jacobi map $\phi : \widetilde{M} \rightarrow M$ between two Jacobi manifolds is same as forward Dirac-Jacobi map when the manifolds are equipped with
corresponding Dirac-Jacobi structures. This follows from the following observation.

Let $\widetilde{M}$ (resp. $M$) be a Jacobi manifold with Jacobi structure $(\widetilde{\pi}, \widetilde{E})$ (resp. $(\pi, E)$) and $\phi : \widetilde{M} \rightarrow M$
be a Jacobi map. Therefore, we have
$$\pi^\sharp = \phi_* \circ \widetilde{\pi}^\sharp \circ \phi^*    ~~~ \text{ and } ~~~ \phi_* \widetilde{E} = E.$$
This implies that
\begin{align*}
 L_{(\pi, E)} =~& \big\{ (\pi^\sharp \alpha + g E, - \langle \alpha, E \rangle) \oplus (\alpha, g) ~|~ (\alpha, g) \in T^*M \times \mathbb{R} \big\} \\
=~& \big\{ (\phi_* \circ \widetilde{\pi}^\sharp \circ \phi^* (\alpha)  + g \phi_* \widetilde{E} ~,~ - \langle \alpha, \phi_* \widetilde{E} \rangle ) \oplus (\alpha, g) ~ |~ (\alpha, g) \in T^*M \times \mathbb{R} \big\} \\
=~& \big\{ (\phi_* \circ \widetilde{\pi}^\sharp \circ \phi^* (\alpha)  + \phi_* ((g \circ \phi)\widetilde{E}) ~,~ - \langle \alpha, \phi_* \widetilde{E} \rangle ) \oplus (\alpha, g) ~ |~ (\alpha, g) \in T^*M \times \mathbb{R} \big\}.
\end{align*}
Let $\widetilde{X} = \widetilde{\pi}^\sharp ( \phi^*\alpha) + (g \circ \phi)\widetilde{E}$ and $f = - \langle \alpha, \phi_* \widetilde{E} \rangle $. Then 
$f \circ \phi =  - \langle \alpha, \phi_* \widetilde{E} \rangle \circ \phi = - \langle \phi^* \alpha, \widetilde{E} \rangle$ and moreover,
$$(\widetilde{X}, f \circ \phi) \oplus (\phi^*\alpha, g \circ \phi) 
= ( \widetilde{\pi}^\sharp ( \phi^*\alpha) + (g \circ \phi)\widetilde{E} ~,~   - \langle \phi^* \alpha, \widetilde{E} \rangle ) \oplus ( \phi^*\alpha, g \circ \phi  ) \in L_{  (\widetilde{\pi}, \widetilde{E})}.$$
Therefore,
$$ L_{(\pi, E)} = \big\{ (\phi_* \widetilde{X}, f ) \oplus (\alpha , g ) ~|~ (\alpha, g) \in T^*M \times \mathbb{R}  \text{ and } (\widetilde{X}, f \circ \phi) \oplus (\phi^*\alpha, g \circ \phi) \in  L_{  (\widetilde{\pi}, \widetilde{E})} \big\} = \phi_* (L_{ (\widetilde{\pi}, \widetilde{E})}).$$
This shows that $\phi$ is a forward Dirac-Jacobi map. It is also easy to verify that if $\phi$ is a Dirac-Jacobi map then $\phi$ is a Jacobi map.

\medskip
The next lemma shows the relation between gauge transformations and push-forward of Dirac-Jacobi structures.
\begin{lemma}\label{gauge-pullback-inv}
 Let $\phi : \widetilde{M} \rightarrow M$ be a smooth map and $L_{\widetilde{M}}$ be a Dirac-Jacobi structure on $\widetilde{M}$. Then for any
$B \in \Omega^1(M)$,
$$\phi_* \big( \tau_{\phi^*B}  L_{\widetilde{M}}  \big)   = \tau_B   \big( \phi_* (L_{\widetilde{M}})  \big).$$
\end{lemma}

\begin{proof}
 We have
\begin{align*}
 \tau_{\phi^*B} L_{\widetilde{M}} =~& \big\{ (\widetilde{X}, \widetilde{f}) \oplus (\widetilde{\alpha}, \widetilde{g}) + i_{(\widetilde{X}, \widetilde{f})} (d \phi^*  B , \phi^* B) ~|~   (\widetilde{X}, \widetilde{f}) \oplus (\widetilde{\alpha}, \widetilde{g}) \in L_{\widetilde{M}} \big\}  \\
=~& \big\{ (\widetilde{X}, \widetilde{f}) \oplus (\widetilde{\alpha} + i_{\widetilde{X}} d \phi^*  B + \widetilde{f} \phi^*B~,~ \widetilde{g} - i_{\widetilde{X}} \phi^*B )  ~|~   (\widetilde{X}, \widetilde{f}) \oplus (\widetilde{\alpha}, \widetilde{g}) \in L_{\widetilde{M}}  \big\}.
\end{align*}
Therefore,
\begin{align*}
 \phi_* \big(   \tau_{\phi^*B} L_{\widetilde{M}} \big) =~& \big\{  ( \phi_* \widetilde{X}, f) \oplus (\alpha, g)   ~ |~  (\widetilde{X}, f \circ \phi) \oplus (\phi^* \alpha, g \circ \phi ) \in \tau_{\phi^*B} L_{\widetilde{M}} \big\} \\
=~& \big\{  ( \phi_* \widetilde{X}, f) \oplus (\alpha, g)   ~ |~    (\widetilde{X}, f \circ \phi) \oplus (\phi^* \alpha - i_{\widetilde{X}} d \phi^*  B - (f \circ \phi) \phi^*B ~,~ g \circ \phi + i_{\widetilde{X}} \phi^*B) \in  L_{\widetilde{M}}  \big\}.
\end{align*}
On the other hand,
\begin{align*}
 \tau_B \big(  \phi_* (L_{\widetilde{M}})  \big) =~& \big\{ (\phi_* \widetilde{X}, f) \oplus (\alpha, g) + i_{(\phi_* \widetilde{X},~ f)} (d B, B) ~|~ (\widetilde{X}, f \circ \phi) \oplus (\phi^*\alpha, g \circ \phi) \in L_{\widetilde{M}} \big\} \\
=~& \big\{ (\phi_* \widetilde{X}, f) \oplus ( \alpha + i_{\phi_* \widetilde{X}} d B + f B~,~ g - i_{\phi_* \widetilde{X}} B)   ~|~ (\widetilde{X}, f \circ \phi) \oplus (\phi^*\alpha , g \circ \phi) \in L_{\widetilde{M}}  \big\} \\
=~& \big\{ (\phi_* \widetilde{X}, f) \oplus (\zeta, h) ~| ~  (\widetilde{X}, f \circ \phi) \oplus (\phi^*\zeta - \phi^* i_{\phi_* \widetilde{X}} d B - \phi^*(f B)~,~ h \circ \phi + (i_{\phi_* \widetilde{X}} B) \circ \phi) \in L_{\widetilde{M}} \big\}.
\end{align*}
Since $i_{\widetilde{X}}\phi^* d B = \phi^* i_{\phi_* \widetilde{X}} d B$~,~ $\phi^* (f B) = (f \circ \phi) \phi^* B$ and $i_{\widetilde{X}} \phi^*B = (i_{\phi_* \widetilde{X}} B) \circ \phi$, we have
$$\phi_* \big( \tau_{\phi^*B}  L_{\widetilde{M}}  \big)   = \tau_B   \big( \phi_* (L_{\widetilde{M}})  \big).$$
\end{proof}

\begin{prop}\label{gauge-contact-groupoid}
Let $(G, \eta)$ be a contact manifold and $\phi : (G, \eta) \rightarrow (M , \pi, E)$ be a Jacobi map. 
If $B \in \Omega^1(M)$ is a $(\pi,E)$-admissible $1$-form on $M$
then $\widehat{\eta} := \eta - \phi^* B$ is a contact $1$-form on $G$.
Moreover, in this case, $\phi: (G, \widehat{\eta}) \rightarrow (M, \tau_B (\pi, E))$ is a Jacobi map.
\end{prop}

\begin{proof}
% The Dirac-Jacobi structure on $G$ associated to the $1$-form $- \widehat{\eta}$ is given by
%$$ L_{- \widehat{\eta}} = L_{- \eta + \phi^* B} = \{ (X,f) + (- i_X d \widehat{\eta} - f \widehat{\eta}~, i_X \widehat{\eta}) ~|~ (X,f) \in TG \times \mathbb{R} \} \subset (TG \times \mathbb{R}) \oplus (T^*G \times \mathbb{R}).$$
From the definition of $L_{- \eta}$, we have $\tau_{\phi^*B} L_{- \eta} = L_{- \eta + \phi^*B} .$
%\begin{align*}
% \tau_{\phi^*B} L_{- \eta} =~& \{   (X,f) \oplus \big( (-i_X d \eta - f \eta, i_X \eta)   + i_{(X,f)} (d \phi^*  B, \phi^* B)        \big) |~ (X,f) \in TG \times \mathbb{R} \}\\
%=~& \{   (X,f) \oplus \big( (-i_X d \eta - f \eta, i_X \eta)   + (i_X d \phi^*  B + f \phi^* B, ~ -i_X \phi^* B)        \big) |~ (X,f) \in TG \times \mathbb{R} \}\\
%=~& \{  (X,f) \oplus ( -i_X d \eta - f \eta +  i_X d \phi^*  B + f \phi^* B  ,   ~i_X \eta    -i_X \phi^* B      )  | ~ (X,f) \in TG \times \mathbb{R}  \} \\
%=~& L_{- \eta + \phi^*B} .
%\end{align*}
Hence,
\begin{align}\label{some-eqn}
\phi_* (L_{- \eta + \phi^*B}) = \phi_* ( \tau_{\phi^*B} L_{-\eta} ) = \tau_B ( \phi_* (L_{-\eta})) =  \tau_B (  L_{(\pi,E)}).
\end{align}
The second equality follows from Lemma \ref{gauge-pullback-inv} and the last equality follows since $\phi$ is a Jacobi map.

Obviously, the Dirac-Jacobi structure $L_{- \eta + \phi^* B}$ satisfies
$$ L_{- \eta + \phi^* B} \cap (\{0\} \oplus (T^*_xG \times \mathbb{R})) = \{0\} ,~~~ \text{ for all } x \in G. $$

From (\ref{some-eqn}) we also have
\begin{align}\label{some-eqn-j}
\phi_* (L_{- \eta + \phi^*B}) =  \tau_B (  L_{(\pi,E)}) = L_{\tau_B(\pi,E)}.
\end{align}
For any $x \in G$, if $(X_x , \lambda) \in ker ~(L_{- \eta + \phi^* B}) |_x$ then $(\phi_* (X_x) , \lambda) \in ker ~(L_{\tau_B (\pi,E)}) |_{\phi(x)} = 0$ as the Dirac-Jacobi
structure $L_{\tau_B (\pi,E)}$ is given by a Jacobi structure. This implies that $X_x \in ker~ \phi_*$ and $\lambda = 0$.

Let $(X_x, 0) \in ker ~( L_{- \eta + \phi^*B} )|_x$.
Then
$$(X_x, 0) \in ker ~\phi_* \cap ker ~(L_{- \eta + \phi^*B})|_x = ker ~\phi_* \cap ker ~(L_{-\eta})|_x = 0.$$
Here the first equality follows since $\phi_* (X_x) = 0$ and the last equality follows since $ker ~ (L_{-\eta})|_x = 0$.
Therefore, the Dirac-Jacobi structure $L_{- \eta + \phi^* B}$ also satisfies
$$ L_{- \eta + \phi^* B} \cap ((T_xG \times \mathbb{R}) \oplus \{0\}) = \{0\},~~~ \text{ for all } x \in G.$$
Hence, by Theorem \ref{contact-dirac-jacobi} the $1$-form $\eta - \phi^* B$ defines a contact structure on $G$. The second part follows from Equation (\ref{some-eqn-j}).
\end{proof}

\begin{remark}\label{some-remark}
Similarly, one can prove the followings.
\begin{itemize}
 \item[(i)] If $\phi : (G, \eta) \rightarrow (M, \pi, E)$ is an anti-Jacobi map, then $\widehat{\eta} = \eta + \phi^* B$ is a contact $1$-form on $G$
and $\phi : (G, \widehat{\eta}) \rightarrow (M, \tau_B(\pi,E))$ is an anti-Jacobi map.
This follows from the following observation that
$$ \phi_* (L_{-\eta - \phi^* B}) = \phi_* (\tau_{-\phi^*B} L_{-\eta}) = \tau_{-B} (\phi_* (L_{-\eta})) = \tau_{-B} ((L_{(\pi,E)})_{-}) = (\tau_{B} (L_{(\pi,E)}))_{-}.$$

 \item[(ii)] If $(\phi, \sigma) : (G, \eta) \rightarrow (M, \pi, E)$ is a conformal Jacobi map, then $\widehat{\eta} = \eta - \sigma \phi^* B$ is a contact $1$-form on $G$
and $(\phi, \sigma) : (G, \widehat{\eta}) \rightarrow (M, \tau_B(\pi,E))$ is a conformal Jacobi map.

Note that the conformal change of a contact form $\eta$ by a nowhere vanishing function $\sigma$ is given by $\frac{\eta}{\sigma}$. Hence, the assertion follows from
the observation that
$$ \phi_* ( L_{ \frac{- \eta + \sigma \phi^* B}{\sigma} }  ) = \phi_* ( L_{- \frac{\eta}{\sigma} + \phi^* B}) = \phi_* (\tau_{\phi^*B} L_{ - \frac{\eta}{\sigma}}) = \tau_B (  \phi_* (L_{- \frac{\eta}{\sigma}})) = \tau_B (L_{(\pi,E)}) .$$
\end{itemize}
\end{remark}

To prove the next theorem, we need the following property of a contact groupoid. More precisely, if $(G \rightrightarrows M, \eta, \sigma)$ is a contact groupoid, the kernels
of $\alpha_*$ and $\beta_*$ are given by
\begin{align}\label{contact-ham}
 (ker~ \alpha_*)|_x = \{ X_{\beta^* f} (x) |~ f \in C^\infty(M) \},~~ (ker~ \beta_*)_x = \{ X_{e^\sigma \alpha^* f}(x) |~  f \in C^\infty(M) \}, \text{ for } x \in G,
\end{align}
where $X_h$ is the hamiltonian vector field on $G$ associated to the function $h \in C^\infty(G)$.

\begin{thm}\label{final-thm}
 Let $(G \rightrightarrows M, \eta, \sigma)$ be a contact groupoid integrating the Jacobi structure $(\pi, E)$ on $M$. Let $B$ 
be a $(\pi,E)$-admissible $1$-form on $M$. Then $(G \rightrightarrows M, ~\eta - e^\sigma \alpha^* B + \beta^* B , \sigma)$ is a contact groupoid integrating
$(M, \tau_B (\pi,E)).$
\end{thm}

\begin{proof}
It follows from Remark \ref{some-remark}(ii) that $\widehat{\eta} = \eta - e^\sigma \alpha^* B$ is a contact $1$-form on $G$ for which $(\alpha, e^\sigma) : (G, \widehat{\eta}) \rightarrow (M, \tau_B (\pi, E))$
is a conformal Jacobi map. Moreover, it follows from (\ref{contact-ham}) that
$$ \beta_* (L_{- \widehat{\eta}}) = \beta_* (L_{- \eta + e^\sigma \alpha^*B}) = \beta_* (\tau_{e^\sigma \alpha^* B} L_{- \eta}) = \beta_* (L_{- \eta}) = (L_{(\pi,E)})_{-}.$$
Therefore, $\beta : (G, \widehat{\eta}) \rightarrow (M, \pi, E)$ is an anti-Jacobi map. Hence,
by Remark \ref{some-remark}(i) the $1$-form
$\eta_B = \eta - e^\sigma \alpha^* B + \beta^* B$ is a contact $1$-form on $G$ and $\beta : (G, \eta_B) \rightarrow (M, \tau_B (\pi, E))$ is an anti-Jacobi map. Similarly,
$(\alpha, e^\sigma) : (G, \eta_B) \rightarrow (M, \tau_B (\pi, E))$ is a conformal Jacobi map.

Moreover,
\begin{align*}
 (\eta_B)_{gh} (X_g \oplus_{TG} Y_h) =~& (\eta + e^\sigma \alpha^* B - \beta^* B)(gh) (X_g \oplus_{TG} Y_h) \\
=~& \eta_g (X_g) + e^{\sigma (g)} \eta_h (Y_h) + e^{\sigma (gh)} B |_{\alpha (gh)} \alpha_* (Y_h) - B|_{\beta (gh)} \beta_* (X_g),
\end{align*}
and
\begin{align*}
(\eta_B)_g (X_g) + e^{\sigma (g)} (\eta_B)_h (Y_h) =~& \eta_g (X_g) + \cancel{e^{\sigma (g)} B |_{\alpha (g)}  \alpha_* (X_g)} - B |_{\beta (g)} \beta_* (X_g) \\
&~ + e^{\sigma (g)} \eta_h (Y_h) + e^{\sigma (g)} e^{\sigma (h)} B |_{\alpha (h)} \alpha_* (Y_h) - \cancel{e^{\sigma (g)} B|_{\beta (h)} \beta_* (Y_h)},
\end{align*}
for all $(X_g, Y_h) \in (TG)^{(2)}.$ It follows that
$$ (\eta_B)_{gh} (X_g \oplus_{TG} Y_h) =  (\eta_B)_g (X_g) + e^{\sigma (g)} (\eta_B)_h (Y_h), \text{ for } (X_g, Y_h) \in (TG)^{(2)}.$$
Hence, $(G \rightrightarrows M, \eta_B, \sigma)$ is a contact groupoid.
Moreover, we have
the map $(\alpha, e^\sigma) : (G, \eta_B) \rightarrow (M, \tau_B (\pi, E))$ is a conformal Jacobi map and $\beta : (G, \eta_B) \rightarrow (M, \tau_B (\pi,E))$
is an anti-Jacobi map. It follows from the uniqueness of the Jacobi structure on the base of a contact groupoid that $(G \rightrightarrows M, \eta_B, \sigma)$
is a contact groupoid integrating $(M, \tau_B(\pi, E)).$
\end{proof}

%\begin{remark}\label{rem-pre}
%For any $1$-form $\eta$ on $M$, the Dirac-Jacobi structure $L_\eta$ as given in (\ref{pre-dj}) is integrable.
%The corresponding precontact groupoid is given by $(M \times \mathbb{R} \times M \rightrightarrows M, ~ \underline{\eta} = e^\sigma pr_3^* \eta - pr_1^* \eta, ~\sigma)$.
%
%For any $B \in \Omega^1(M)$, the gauge transformation Dirac-Jacobi structure on the base is given by $\tau_B (L_\eta) = L_{\eta + B}$.  Therefore, by Theorem
%\ref{final-thm-pre} the corresponding precontact $1$-form on the groupoid is given by 
%$$\underline{\eta} + e^\sigma pr_3^* B - pr_1^* B =  e^\sigma pr_3^* (\eta + B) - pr_1^* (\eta + B).$$
%This is precisely the global precontact $1$-form $\underline{\eta + B}$ on the groupoid associated to the Dirac-Jacobi structure $L_{\eta + B}$ on the base.
%\end{remark}

%Next, we turn our attention to the correspondence between contact groupoids and Jacobi structures.

%\begin{defn}
%A contact groupoid is a precontact groupoid $(G \rightrightarrows M,  \eta, \sigma)$ such that $\eta$ is a contact $1$-form.
%\end{defn}

\begin{remark}
Let $\eta$ be a contact $1$-form on $M$ considered as a Jacobi structure. Then it is integrable and the corresponding contact groupoid is given by
$(M \times \mathbb{R} \times M \rightrightarrows M , ~\underline{\eta} = e^\sigma pr_3^*\eta - pr_1^* \eta, ~\sigma)$. The source, target and the partial multiplication of the groupoid structure is given by 
\begin{align*}
\alpha (x, t, y) =&~ pr_3 (x, t , y) = y,\\
\beta (x, t , y) =&~ pr_1 (x, t, y) = x,\\
(x, t, y) (y , s , z ) =&~ (x, t + s , z),
\end{align*}
for $(x, t, y) , (y , s , z) \in M \times \mathbb{R} \times M$. The multiplicative function $\sigma$ on this groupoid is just the projection onto the second factor.

If $B \in \Omega^1(M)$ is a $1$-form on the base $M$ such that the gauge transformation defines a Jacobi structure (this is infact a contact structure), then we have
$\tau_B (\eta) = \eta - B$. By Theorem \ref{final-thm} the corresponding contact $1$-form on the groupoid is given by
\begin{align*}
 \underline{\eta} - e^\sigma pr_3^* B + pr_1^* B = e^\sigma pr_3^* (\eta - B) - pr_1^* (\eta - B).
\end{align*}
This is precisely the global contact $1$-form $\underline{\eta- B}$ on the groupoid associated to the contact structure $\tau_B(\eta) = \eta - B$ on the base.
\end{remark}

In the next section we give an alternative proof of Theorem \ref{final-thm} using gauge transformations of multiplicative Jacobi structures (see Theorem \ref{final-theorem} and Remark \ref{final-theorem-rem}).

\medskip

\section{Gauge transformations of Jacobi groupoids}\label{sec-6}

In this section, we study gauge transformations of Jacobi groupoids and of generalized Lie bialgebroids.

We recall that a Jacobi algebroid is a pair $(A, \phi_0)$ of a Lie algebroid $A$ together with a $1$-cocycle $\phi_0 \in \Gamma A^*$ of it \cite{igl-marr, grab-mar}. Given a Jacobi algebroid $(A, \phi_0)$ the differential $d_A$ of the Lie algebroid $A$ can be twisted by $\phi_0$ to define a new differential
$$ d_A^{\phi_0} : \Gamma (\largewedge^\bullet A^*) \rightarrow  \Gamma (\largewedge^{\bullet +1} A^*), ~ \alpha \mapsto d_A \alpha + \phi_0 \wedge \alpha .$$
Moreover, the Gerstenhaber bracket $[-,-]$ on the space of multisections of $A$ can be twisted by $\phi_0$ to define a new bracket (called Schouten-Jacobi bracket) 
$[-,-]^{\phi_0} : \Gamma(\largewedge^\bullet A) \times \Gamma(\largewedge^\bullet A) \rightarrow \Gamma(\largewedge^{\bullet} A)$ of degree $-1$ by the following
$$[P, Q]^{\phi_0} = [P, Q] + (-1)^{p+1}(p-1) P \wedge i_{\phi_0} Q - (q-1) i_{\phi_0} P \wedge Q, ~\text{ for } P \in \Gamma(\largewedge^p A), Q \in \Gamma(\largewedge^q A).$$

\begin{defn}\cite{igl-marr}
A generalized Lie bialgebroid $((A, \phi_0), (A^*, X_0))$ consists of a pair of Jacobi algebroids $(A, \phi_0)$ and $(A^*, X_0)$ in duality satisfying
$$ d_*^{X_0} [P, Q]^{\phi_0} = [d_*^{X_0} P~,~ Q]^{\phi_0}  +~ (-1)^{p+1}   [P ~,~ d_*^{X_0} Q]^{\phi_0},$$
for $P \in \Gamma(\largewedge^p A)$ and $Q \in \Gamma(\largewedge^\bullet A)$. Here $d_*^{X_0}$ denote the differential of the Jacobi algebroid $(A^*, X_0)$ and
$[-,-]^{\phi_0}$ denotes the Schouten-Jacobi bracket on the multisections of $A$ associated the Jacobi algebroid $(A, \phi_0).$
\end{defn}

See \cite{igl-marr2, grab-mar , das0} for more details. If $\phi_0 = 0$ and $X_0 = 0$, one reduces the definition of a Lie bialgebroid \cite{mac-xu}.

\begin{exam}
If $(M, \pi, E)$ is a Jacobi manifold, the pair $\big( ~(T^*M \times \mathbb{R}, (-E, 0)) , ~(TM \times \mathbb{R}, (0,1)) ~\big)$ is a 
generalized Lie bialgebroid \cite{igl-marr}. 
\end{exam}

Conversely, the base of a generalized Lie bialgebroid carries a Jacobi structure.
Let $((A, \phi_0), (A^*, X_0))$ be a generalized Lie bialgebroid over $M$.
Suppose the bracket and anchor of the Lie algebroids $A$ and $A^*$ are given by $([-,-], a)$ and $([-,-]_*, a_*)$ , respectively.
Then the induced Jacobi structure $(\pi, E)$ on $M$ is given by
$$ (\pi, E)^\sharp := \big(a(-), \phi_0(-)\big) \circ \big(a_* (-), X_0 (-)\big)^* : T^*M \times \mathbb{R} \longrightarrow TM \times \mathbb{R} ,$$
where the bundle maps $(a(-), \phi_0 (-)) : A \longrightarrow TM \times \mathbb{R}$ 
and $(a_* (-), X_0 (-)) : A^* \longrightarrow TM \times \mathbb{R}$ are respectively given by
$X \mapsto (a(X), \phi_0 (X))$ and $\alpha \mapsto (a_* (\alpha), X_0 (\alpha))$, for $X \in \Gamma A,~\alpha \in \Gamma A^*$ \cite{igl-marr}.

Let $B \in \Omega^1(M)$ be a $(\pi, E)$-admissible $1$-form on $M$. One can easily check that this condition is being equivalent to the invertability of the map
$$ \psi_B :=~  Id ~ + ~ \big(a(-), \phi_0(-)\big)^* \circ \widetilde{(d B, B)} \circ \big(a_* (-), X_0 (-)\big) : A^* \longrightarrow A^*.$$
Hence, one can define a new Lie algebroid structure on $A^*$ whose bracket and anchor are given by
$$ [\alpha, \beta]_*^B := \psi_B ~ [\psi_B^{-1} (\alpha), \psi_B^{-1} (\beta)]_* ~~~, \quad a_*^B := a_* \circ \psi_B^{-1}.$$
We denote this Lie algebroid structure on $A^*$ by $(A^*)^B$. Moreover,
$X_0^B (-)  := X_0 \circ \psi_B^{-1}(-)$ defines a $1$-cocycle of this Lie algebroid. Actually, we have
$$(a_*^B (-), X_0^B(-)) = (a_* (-), X_0(-)) \circ \psi_B^{-1} (-).$$
 We denote the Jacobi algebroid $((A^*)^B, X_0^B )$ simply by $(A^*, X_0)_B$. Moreover, the pair $( (A, \phi_0), (A^*, X_0)_B )$ is a generalized Lie bialgebroid. We call this generalized Lie bialgebroid as a gauge transformation of the given one. 
Note that the bundle map associated to the Jacobi structure on $M$ induced from the
generalized Lie bialgebroid $((A, \phi_0), (A^*, X_0)_B)$ is given by
\begin{align*}
  \big( a(-), \phi_0 (-) \big) \circ \big(a_*^B (-), X_0^B (-)\big)^* =  \big( a(-), \phi_0 (-) \big) \circ (\psi_B^*)^{-1} \circ (a_* (-), X_0 (-))^*
\end{align*}
%$$ \big( a(-), \phi_0 (-) \big) \circ \big(a_*^B (-), X_0^B (-)\big)^* : T^*M \times \mathbb{R} \longrightarrow TM \times \mathbb{R}.$$
It is easy to show that this map coincides with $(\pi, E)^\sharp \circ \big(Id + \widetilde{(d B, B)} \circ (\pi, E)^\sharp \big)^{-1}$
which is same as $(\tau_B (\pi,E))^\sharp$. Hence, the
Jacobi structure on $M$ induced from the  generalized Lie bialgebroid $((A, \phi_0), (A^*, X_0)_B)$ is simply given by $\tau_B (\pi, E).$

\begin{remark}\label{glb-remark}
Let $(\pi, E)$ be a Jacobi structure on $M$ and consider the generalized Lie bialgebroid $\big(~(T^*M \times \mathbb{R}, (-E, 0)) , (TM \times \mathbb{R}, (0,1)) ~ ~~\big)$. 
Then one can check that the gauge transformed generalized Lie bialgebroid $\big(~(T^*M \times \mathbb{R}, (-E, 0)), (TM \times \mathbb{R}, (0,1))_B~\big)$ is isomorphic to the
generalized Lie bialgebroid $\big(~(T^*M \times \mathbb{R}, (-E_B, 0)) , (TM \times \mathbb{R}, (0,1)) ~ ~~\big)$ associated to the transformed Jacobi
structure $\tau_B(\pi, E) = (\pi_B, E_B)$ on $M$.
\end{remark}

\medskip

Next, we recall multiplicative Jacobi structures on Lie groupoids and study gauge transformations of them in the multiplicative sense.
First, let us recall few things.
Let $G \rightrightarrows M$ be a Lie groupoid and $\sigma \in C^\infty (G)$ be a multiplicative function. Then the tangent Lie groupoid $TG \rightrightarrows TM$ can be twisted by $\sigma$ to define a new Lie groupoid $TG \times \mathbb{R} \rightrightarrows TM \times \mathbb{R}$. The source, target and partial multiplication are given by
\begin{align*}
( \alpha_* )_\sigma  (X_g , \lambda) =& ~(\alpha_* (X_g), X_g (\sigma) + \lambda )\\
( \beta_* )_\sigma (Y_h, \mu) =& ~( \beta_* (Y_h), \mu )\\
(X_g , \lambda) \oplus_{TG \times \mathbb{R}} (Y_h, \mu) =& ~ (X_g \oplus_{TG} Y_h, \lambda),
\end{align*}
for $(X_g , \lambda) \in T_gG \times \mathbb{R}, ~ (Y_h, \mu) \in T_hG \times \mathbb{R}$ and $(\alpha_* )_\sigma  (X_g , \lambda) = (\beta_*)_\sigma (Y_h, \mu)$.

Let $A \rightarrow M$ be the Lie algebroid of $G \rightrightarrows M$. One can also twist the usual cotangent groupoid $T^*G \rightrightarrows A^*$ by the multiplicative function $\sigma$ to define a new groupoid $T^*G \times \mathbb{R} \rightrightarrows A^*$ with structure maps
\begin{align*}
(\tilde \alpha)_\sigma (\omega_g, \gamma) =&~ e^{- \sigma (g)} \tilde{\alpha} (\omega_g)~ \\
(\tilde \beta)_\sigma (\nu_h, \zeta) =&~ \tilde{\beta} (\nu_h)  -  \zeta (d \sigma)|_{\tilde{\beta} (h)}\\
(\omega_g, \gamma) \oplus_{T^*G \times \mathbb{R}} (\nu_h, \zeta) = &~ \big( \omega_g + e^{\sigma (g)} \zeta (d \sigma)|_g  \oplus_{T^*G} ~ e^{\sigma (g)} \nu_h ~,~ \gamma + e^{\sigma (g)} \zeta  \big) ,
\end{align*}
for $(\omega_g, \gamma) \in T_g^*G \times \mathbb{R}, ~ (\nu_h, \zeta) \in T_h^*G \times \mathbb{R}$ and $(\tilde \alpha)_\sigma (\omega_g, \gamma) = (\tilde \beta)_\sigma (\nu_h, \zeta)$.
This Lie groupoid is called the $1$-jet Lie groupoid of $G$ twisted by $\sigma$ \cite{igl-marr2}.

\begin{defn}\cite{igl-marr2}
A Jacobi groupoid is a Lie groupoid $G \rightrightarrows M$ together with a multiplicative function $\sigma \in C^\infty(G)$ and a Jacobi structure $(\pi_G, E_G)$ on $G$ such that the induced map $(\pi_G, E_G)^\sharp : T^*G \times \mathbb{R} \rightarrow TG \times \mathbb{R}$ is a Lie groupoid morphism 
\[
\xymatrixrowsep{0.4in}
\xymatrixcolsep{0.6in}
\xymatrix{
	T^*G \times \mathbb{R} \ar[r]^{(\pi_G, E_G)^\sharp} \ar[d] \ar@<-4pt>[d] & TG \times \mathbb{R} \ar[d] \ar@<-4pt>[d]\\
	A^*  \ar[r]_{} & TM \times \mathbb{R}
}
\]
from the $1$-jet Lie groupoid to the twisted tangent Lie groupoid defined above. 
\end{defn}

Let $(G \rightrightarrows M,  \pi_G, E_G , \sigma)$ be a Jacobi groupoid as above with the Lie algebroid $A$. Then by differentiating $\sigma$ we get a $1$-cocycle $\phi_0 \in \Gamma A^*$ of the Lie algebroid. 
Note that, the dual bundle $A^*$ can be identified with the conormal bundle
$(TM)^0 \rightarrow M.$ Using this identification, one gets a Lie algebroid structure on $A^*$ whose bracket $[-,-]_*$ and the anchor $a_*$ are given by
$$ [\alpha, \beta]_* (x) = pr_1 \big([(\widetilde{\alpha}, 0), (\widetilde{\beta}, 0)]_{(\pi_G, E_G)} \big)(x) ~ \text{ and } ~ a_* (\alpha)(x) = \pi_G^\sharp (\widetilde{\alpha})(x), ~ \text{ for } \alpha, \beta \in \Gamma A^*,~ x \in M,$$
where $\widetilde{\alpha}, \widetilde{\beta}$ be any extension of $\alpha$ and $\beta$ to $1$-forms on $G$. 
Moreover, there is a distinguised $1$-cocycle $X_0 \in \Gamma A$ of this Lie algebroid given by
$\langle X_0 (x) , \alpha (x) \rangle =  - \langle \alpha(x) , E_G (x) \rangle$, for $x \in M$ and $\alpha (x) \in A_x^* \cong (T_xM)^0$. The pair
$((A, \phi_0), (A^*, X_0))$ forms a generalized Lie bialgebroid 
\cite{igl-marr2}. Note that the structures on $(A^*, X_0)$ depends only on the Jacobi structure $(\pi_G, E_G)$ on $G$. It also turns out that the base map of the Lie groupoid morphism $(\pi_G, E_G)^\sharp : T^*G \times \mathbb{R} \rightarrow TG \times \mathbb{R}$ is given by $(a_*(-), X_0 (-) : A^* \rightarrow TM \times \mathbb{R}.$

Let $(G \rightrightarrows M,  \pi_G, E_G , \sigma)$ be a Jacobi groupoid with its generalized Lie bialgebroid $((A, \phi_0), (A^*, X_0))$. Since a gauge transformation of the generalized Lie bialgebroid $((A, \phi_0), (A^*, X_0))$
effects only on $(A^*, X_0)$, its effect on the associated Jacobi groupoid is only a change of the Jacobi structure. More precisely, we have the following result.

\begin{thm}\label{final-theorem}
 Let $(G \rightrightarrows M, \pi_G, E_G , \sigma)$ be a Jacobi groupoid with generalized Lie bialgebroid $((A, \phi_0), (A^*, X_0))$ and the induced Jacobi structure
$(\pi, E)$ on $M$. Let $B \in \Omega^1(M)$ and let $B_G = e^\sigma \alpha^* B - \beta^* B \in \Omega^1(G)$. If
$B$ is $(\pi, E)$-admissible then $B_G$ is $(\pi_G, E_G)$-admissible and in this case,
$(G \rightrightarrows M, \tau_{B_G} (\pi_G, E_G), \sigma)$ is a Jacobi groupoid with its generalized Lie bialgebroid $((A, \phi_0), (A^*, X_0)_B).$
\end{thm}

To prove this theorem, we need the following lemma. The proof is straight-forward.

\begin{lemma}
 Let $G \rightrightarrows M$ be a Lie groupoid and $\sigma \in C^\infty(G)$ a multiplicative function. Let $B$ be any $1$-form on $M$
and take $B_G = e^\sigma \alpha^* B - \beta^* B \in \Omega^1(G)$. Then 
$\widetilde{(d B_G , B_G)} : TG \times \mathbb{R} \rightarrow T^*G \times \mathbb{R}$ is a Lie groupoid morphism.
\end{lemma}

\noindent {\em Proof of Theorem \ref{final-theorem}.} 
One can easily verify that the map
$$ (\beta_* , Id) : ker~ \big(\tau_{B_G} (L_{(\pi_G, E_G)}) \big)|_g \rightarrow ker~ \big(  \tau_B (L_{(\pi,E)})   \big)|_{\beta (g)}, ~ (X_g , \lambda) \mapsto (\beta_* X_g , \lambda)$$
defines an isomorphism between the kernels of the respective Dirac-Jacobi structures. Since $B$ is $(\pi,E)$-admissible, we have $ker ~ \big(  \tau_B (L_{(\pi,E)})   \big) = 0$. Therefore,
$ker~ \big(\tau_{B_G} (L_{(\pi_G, E_G)}) \big) = 0$. Hence, the Dirac-Jacobi structure $\tau_{B_G} (L_{(\pi_G, E_G)})$ is given by the graph of a Jacobi structure. In other words,
$B_G$ is $(\pi_G, E_G)$-admissible.

Since $(\pi_G, E_G)^\sharp : T^*G \times \mathbb{R}  \longrightarrow TG \times \mathbb{R}$ and $\widetilde{(d B_G, B_G)} : TG \times \mathbb{R} \longrightarrow T^*G \times \mathbb{R}$ are groupoid morphism, the composition
$$ \big(  \tau_{B_G} (\pi_G, E_G) \big)^\sharp = (\pi_G, E_G)^\sharp \big(  Id ~+~ \widetilde{(d B_G , B_G)} \circ (\pi_G, E_G)^\sharp  \big)^{-1} =  (\pi_G, E_G)^\sharp \circ \Phi_G^{-1} : T^*G \times \mathbb{R}  \longrightarrow TG \times \mathbb{R}$$
is also a Lie groupoid morphism, where $\Phi_G =  Id + \widetilde{(d B_G , B_G)} \circ (\pi_G, E_G)^\sharp $ is the invertible bundle map. Hence, $(G \rightrightarrows M,  \tau_{B_G} (\pi_G, E_G) , \sigma )$ is a Jacobi groupoid.
Moreover, it follows from the above expression of the map $\big(  \tau_{B_G} (\pi_G, E_G) \big)^\sharp$ that
\begin{align*}
 \big(  \tau_{B_G} (\pi_G, E_G) \big)^\sharp |_{A^*} 
=&~ \big(a_* (-), X_0 (-)\big) \big(   Id ~ + ~ (a(-), \phi_0(-))^* \circ \widetilde{(d B, B)} \circ \big(a_* (-), X_0 (-)\big) \big)^{-1} \\
=&~ \big(a_* (-), X_0 (-)\big) \circ \psi_B^{-1}
= (a_*^B (-), X_0^B(-)).
\end{align*}
Finally, the Lie bracket of $\alpha, \beta \in \Gamma A^*$ induced from the Jacobi groupoid  $(G \rightrightarrows M,  \tau_{B_G} (\pi_G, E_G) , \sigma )$ is given by
\begin{align*}
pr_1 \big(~[(\widetilde{\alpha}, 0), (\widetilde{\beta}, 0)]_{\tau_{B_G}(\pi_G, E_G)}~ \big)|_{M}
=&~ pr_1 \big(  \Phi_G [~ \Phi_G^{-1} (\widetilde{\alpha}, 0) , \Phi_G^{-1} (\widetilde{\beta}, 0) ~ ]_{(\pi_G, E_G)} \big)\big|_{M} ~~~ \text{~(by Prop \ref{gauge-iso-lie})}\\
=&~ \psi_B  \big(  [\Phi_G^{-1} (\widetilde{\alpha}, 0), \Phi_G^{-1} (\widetilde{\beta}, 0) ]_{(\pi_G, E_G)}  \big)|_M \\
=&~ \psi_B \big( [ ~(~\widetilde{\psi_B^{-1} (\alpha)} , 0) ,  (~\widetilde{\psi_B^{-1} (\alpha)} , 0) ~ ]_{(\pi_G, E_G)}  \big)\big|_{M} \\
=&~ \psi_B [~\psi_B^{-1} (\alpha), \psi_B^{-1} (\beta)~]_* = [\alpha, \beta]_*^B.
\end{align*}
Hence, the result follows.

\begin{remark}\label{final-theorem-rem}
Note that Theorem \ref{final-thm} follows as a corollary of Theorem \ref{final-theorem} and Remark \ref{glb-remark}. Let $(G \rightrightarrows M, \eta, \sigma)$
be a contact groupoid with the Jacobi structure $(\pi,E)$ on $M$. Think the contact groupoid as a Jacobi groupoid with the Jacobi structure on $G$ induced by
$\eta$. Then its generalized Lie bialgebroid is given by $\big( (T^*M \times \mathbb{R}, (-E, 0)), (TM \times \mathbb{R}, (0,1))  \big)$. Let $B$ be a 
$(\pi,E)$-admissible $1$-form on $M$ with transformed Jacobi structure $\tau_B (\pi,E) = (\pi_B, E_B)$. Then by Remark \ref{glb-remark} the transformed generalized Lie bialgebroid
is isomorphic to $\big( (T^*M \times \mathbb{R}, (-E_B, 0)), (TM \times \mathbb{R}, (0,1))  \big)$. Therefore, by Theorem \ref{final-theorem} we have
$(G \rightrightarrows M, \tau_{B_G} (\eta), \sigma)$ is a contact groupoid integrating the generalized Lie bialgebroid $\big( (T^*M \times \mathbb{R}, (-E_B, 0)), (TM \times \mathbb{R}, (0,1))  \big)$.
In other words, $(G \rightrightarrows M, \eta - B_G, \sigma)$ is a contact groupoid integrating the Jacobi structure $\tau_B (\pi,E)$. Hence 
Theorem \ref{final-thm} follows.
\end{remark}

\medskip

\section{$B$-field transformations of generalized contact structures}\label{sec-7}
Generalized contact structures 
are odd analouge of generalized complex structures and generalization of contact structures \cite{wade2}.
A line bundle approach of this notion was studied in \cite{vitag-wade, jonas-luca}.
In this final section, we study some symmetries of generalized contact structures, namely, $B$-field transformations. 

\begin{defn}
 Let $M$ be a manifold of dimension $2n + 1$. A generalized contact structure on $M$ is a bundle map
$$\mathcal{I} : \mathcal{E}^1(M) \longrightarrow \mathcal{E}^1(M)$$
satisfying
\begin{center}
$\mathcal{I}^2 = Id, \qquad \qquad ~~ \mathcal{I}^* = - \mathcal{I} \qquad \text{ and } \qquad \mathcal{N}_{\mathcal{I}} = 0.$
\end{center}
Here $\mathcal{I}^*$ denote the adjoint of $\mathcal{I}$ with respect to non-degenerate pairing $\langle \langle -, - \rangle \rangle$ defined in (\ref{pairing}) and $\mathcal{N}_{\mathcal{I}}$
denote the Nijenhuis torsion of $\mathcal{I}$ with respect to the generalized Dorfman bracket $\llbracket-,-\rrbracket$. A manifold equipped with a generalized
contact structure is called a generalized contact manifold.
\end{defn}

\begin{exam}
Any contact manifold is a generalized contact manifold. In particular, if $\eta$ is a contact $1$-form on $M$ with Jacobi structure $(\pi,E)$, then as a matrix block
\begin{align}\label{matrix}
\mathcal{I} = \left( \begin{array}{cc}
0     &      (\pi,E)^\sharp\\
\widetilde{(d \eta, \eta)}     &       0
\end{array}
\right)
\end{align}
is a generalized contact structure on $M$.
\end{exam}

Motivated from the $B$-field transformations of generalized complex structures, here 
we define a similar type transformation for generalized contact structures.

% and discuss its relationship with gauge transformation of contact structures.

Let $\mathcal{I}$ be a generalized contact structure on $M$. For any $1$-form $B$, consider the orthogonal automorphism $exp (B)$ of the bundle
$\mathcal{E}^1(M)$ via
\begin{align*}
exp (B) = \left( \begin{array}{cc}
Id     &      0\\
\widetilde{(d B, B)}     &       Id
\end{array}
\right).
\end{align*}
Then it is straight forward to verify that $\tau_B(\mathcal{I}) = exp (B) \circ \mathcal{I} \circ exp (-B)$ is another generalized contact structure on $M$.
This follows from the fact that $exp (B)$ preserves the generalized Dorfman bracket on $\mathcal{E}^1(M).$
The generalized contact structure $\tau_B(\mathcal{I})$ is called the $B$-field transformation of $\mathcal{I}$.

\begin{remark}
When the generalized contact structure $\mathcal{I}$ is given by a contact form $\eta$ as in (\ref{matrix}), the $B$-field transformation
$\tau_B(\mathcal{I})$ is given by
\begin{align*}
\tau_B(\mathcal{I}) = \left( \begin{array}{cccc}
- (\pi,E)^\sharp \widetilde{(d B , B)}      & &&  ~~~(\pi,E)^\sharp \\
   \widetilde{(d \eta, \eta)} - \widetilde{(d B, B)} (\pi,E)^\sharp  \widetilde{(d B, B)}    &  && ~~~  \widetilde{(d B , B)} (\pi,E)^\sharp
\end{array}
\right).
\end{align*}
Obviously, this is not given by any contact form. Therefore, $B$-field transformations of a contact form (considered as a generalized contact structure)
need not be contact. This holds if and only if $B = 0$.
\end{remark}

\medskip

%\mbox{ }\\
%
%\providecommand{\bysame}{\leavevmode\hbox to3em{\hrulefill}\thinspace}
%\providecommand{\MR}{\relax\ifhmode\unskip\space\fi MR }
% \MRhref is called by the amsart/book/proc definition of \MR.
%\providecommand{\MRhref}[2]{%
 % \href{http://www.ams.org/mathscinet-getitem?mr=#1}{#2}
%}
%\providecommand{\href}[2]{#2}

\end{document}